\documentclass[preprint,11pt,authoryear]{elsarticle}

\usepackage[utf8]{inputenc}
\usepackage[english]{babel}
\usepackage{url}

\usepackage{amsmath,amssymb}

\usepackage{natbib}
\usepackage{graphicx}
\usepackage{color}
\usepackage{ltablex}
\usepackage{longtable}
\usepackage{subcaption}
\usepackage{float}
\usepackage{amsthm}
\newtheorem{lemma}{Lemma}
\newtheorem{theorem}{Theorem}
\newtheorem{statement}{Statement}

\usepackage{lineno} %hyperref
\modulolinenumbers[5]

\graphicspath{{./figures/}}

\newcommand{\bbE}{\mathrm{E}}
\newcommand{\bbV}{\mathrm{Var}}
\newcommand{\vecBeta}{\beta}
\newcommand{\vecGamma}{\gamma}
\newcommand{\sigu}{\sigma_u}
\newcommand{\sigv}{\sigma_v}
\newcommand{\sigs}{\sigma_*}

\journal{European journal of Operational Research}

\begin{document}

\begin{frontmatter}

\title{Technical efficiency and inefficiency: \\ SFA misspecification}

%% Group authors per affiliation:
\author{S.~Kumbhakar\fnref{myfootnote}}
\address{State University of New York at Binghamton, USA, kkar@binghamton.edu}
\fntext[myfootnote]{Corresponding author}

\author{A.~Peresetsky}
\address{National Research University Higher School of Economics, Moscow, Russia}

\author{Y.~Shchetynin}
\address{AO Kaspersky Lab, Moscow, Russia, evgeniy.schetinin@gmail.com}

\author{A.~Zaytsev}
\address{Skolkovo Institute of Science and Technology, Moscow, Russia, a.zaytsev@skoltech.ru}

\begin{abstract}
%The effect of external factors $z$ on technical  inefficiency $u$ in %stochastic frontier (SF) production models is often specified through %the variance of $u$.
%In applied papers it is usually assumed that if the variance of $u$ %is positively related to $z$, then an increase in $z$ leads to an %increase in technical inefficiency measured by $TI = \bbE (u)$, and a %decrease in technical efficiency $TE = \bbE(e^{-u})$. Thus it is %assumed that of the marginal effects of $z$ on $TI$ and  $TE$ are %opposite.
%We provide a formal proof of this for the case of normally %distributed random error ($v$) and exponentially and half-normally %distributed $u$.  
%The proof is provided both for the unconditional means $\bbE (u)$, %$\bbE (e^{-u})$, and for 
%the conditional means $\bbE (u| (v - u))$ and $\bbE(e^{-u}|(v - u))$.

The effect of external factors $z$ on technical inefficiency ($TI$) in stochastic frontier (SF) production models is often specified through the variance of inefficiency term $u$.
In this setup the signs of marginal effects of $z$ on $TI$ and technical
efficiency $TE$ identify how one should control $z$ to increase $TI$ or decrease $TE$.
We prove that these signs for $TI$ and $TE$ are opposite for typical setups with normally distributed random error $v$ and exponentially or half-normally distributed $u$ for both conditional and unconditional case.

On the other hand, we give an example to show that signs of the marginal effects of $z$ on $TI$ and  $TE$ may coincide, at least for some ranges of $z$. In our example, the distribution of $u$ is a mixture of two distributions, and the proportion of the mixture is a function of $z$.  Thus if the real data come from this mixture distribution, and we estimate model parameters with an exponential or half-normal distribution for $u$, the estimated efficiency and the marginal effect of $z$ on $TE$ would be wrong. 
Moreover for a misspecified model, the rank correlations between the true and the estimated values of TE could be small and even negative for some subsamples of data.
These results are demonstrated by simulations.

% The effect of external factors $z$ on technical  inefficiency ($TI$) in stochastic frontier (SF) production models is often specified through the variance of inefficiency term $u$.
% In this setup the marginal effects of $z$ on $TI$ and technical
% efficiency $TE$ are of opposite signs.
% We provide a formal proof of this for the case of normally distributed random error ($v$) and exponentially and half-normally distributed $u$.  
% The proof is provided both conditional and unconditional $TI$ and $TE$.

% We also provide an example to show that signs of the marginal effects of $z$ on $TI$ and  $TE$ may 
% coincide
% at least, for some ranges of $z$, when the distribution of $u$ is a mixture of two distributions and the proportion of mixture is a function of $z$.  Thus if the real data on $u$ comes from a mixture distribution with either the mean and/or the proportions are functions of $z$, and we estimate model parameters with an exponential or half-normal distribution for $u$, the estimated efficiency and the marginal effect of $z$ on $TE$ would be wrong. 
% In that case of model misspecification rank correlation between true and estimated values of TE could be small and even negative for some subsamples of data.
% These results are demonstrated by simulations.
\end{abstract}


\end{frontmatter}

% \linenumbers

\section{Introduction}

%The technical efficiency is the ratio of actual output to the maximum possible output.
%Stochastic frontier analysis (SFA) is one of the most popular tools in Econometrics for analysis of the technical efficiency (\cite{aigner1977formulation,meeusen1977efficiency}).

%%(\cite{wang2002dependence, kumbhakar-sun2013})

Stochastic frontier (SF) production model \citep{aigner1977formulation,meeusen1977efficiency} is designed to estimate the observation-specific technical inefficiency $TI$. 
The model has two separate error terms: a symmetrical statistical noise $v$ and a non-negative error term $u$ that represents the technical inefficiency. 
The complete specification of SF model also includes specification of distributions for $v$ and $u$.
If $v$ has a normal distribution, and $u$ has an exponential distribution, then the SF model is called normal-exponential,
if $v$ has a normal distribution, and $u$ has a half-normal distribution, then the SF model is called normal-half normal.
To accommodate determinants of inefficiency $z$, the SF model is generalized to make $u$ heteroscedastic (\cite{kumbhakar2000stochastic, wang2003stochastic}, among many others). 

Our goal is to investigate marginal effects of $z$ on $TI$ as well as  technical efficiency ($TE$) for the normal-exponential and normal-half normal models. We assume $u$ to be heteroscedastic, i.e., the variance of $u$ is a function of $z$. 
Suppose that an increase in $z$ leads to an increase in $TI$ measured as $\bbE (u)$ or $\bbE (u|(v - u))$. 
Does it mean that $TE$ measured as $TE = \bbE (e^{-u})$ or $TE = \bbE (e^{-u}|(v - u))$ (see \cite{battese1988prediction}) will decrease? 
Althow it is inuitive
%Implicitly it is assumed in many papers. 
to the best of our knowledge there is no formal proof of this in the literature.
We provide a proof of this statement for the conditional means for two  exponential and half-normal distribution of $u$. 

A number of papers in the past have considered similar issues. For example,
\cite{wang2002dependence},  \cite{ray2015benchmarking}
%\citep{wang2002dependence, ray2015benchmarking} 
derived an expression for marginal effects of the $z$ variables on the expected value of inefficiency $\bbE (u)$ and showed that the sign of the sign of the marginal effects of $z$ is determined by the sign of $z$ in the variance function of $u$. 
%In~\cite{bera1999estimating} authors evaluated monotonicity of the dependence of the technical efficiency on the overall error of estimate of the technical efficiency.
%\cite{kumbhakar2000stochastic}  estimated marginal effects of the technical efficiency with respect to heteroscedasticity factors.
\cite{kumbhakar-sun2013} derived formulas for the marginal effect of exogenous factors on the 
observation-specific inefficiency $\bbE (u|(v - u))$ for the normal-truncated normal model with heteroscedasticity in both  $v$ and $u$. They demonstrated that for this model signs of the marginal effect may vary across observations.

In addition to the   stochastic frontier model with exponential or half-normal distribution of the inefficiency term we consider a model with a discrete distribution of the inefficiency term.
Properties of these models can differ from the properties of the commonly used SF models \citep{kumbhakar2000stochastic}. 
First, for such models an increase in $z$ may increase both $TI$ and $TE$, which is not possible in the usual normal-exponential  and normal-half normal models.
It means that if the true model for $u$ is the discrete model then applying the usual normal-exponential 
% ({\color{red} or normal-half normal})
% ({\color{blue} 
% we have no proof for the normal-half normal -- I suggest to skip this)
% }
model may result in wrong conclusions on the directions of the marginal effects of the $z$ variables on $TE$ of the production units.
Also, it may result in incorrect rankings of the production units by their estimated $TE$. 
More generally the ranking of the production units by their estimated $TE$ might be different from their rankings in terms of their "true" $TE$. 

The impact of the model misspecification on the estimated TE is was studied, using simulations, among other papers in  \cite{yu1998effects, ruggiero1999efficiency,
ondrich2001efficiency, andor2017pseudolikelihood, andor2019combining}. \cite{ruggiero1999efficiency} concluded, that if data are generated by normal-half normal model, then TE estimates by true (normal-half normal) and misspesified (normal-exponential) models provide similar results. Thus this type of misspecification in incorrect choice of the error distribution is not problematic. 
Some papers~\citep{yu1998effects,ruggiero1999efficiency,ondrich2001efficiency} use rank correlation between true and estimated values of TE as a measure of the model misspecification. Another papers~\citep{andor2017pseudolikelihood,andor2019combining} use RMSE measure as the distance between true and estimated TE for performance comparison of different models.
\cite{giannakas2003predicting} demonstrated that predictions of TE are
sensitive to the misspecification of the functional form of the production function
in stochastic frontier regression. 

The rest of the paper is organized as follows. In Section~\ref{sec:limit_effects} we introduce the normal-exponential and normal-half normal model and derive the formulas for computing the marginal effects of determinants of technical efficiency and technical inefficiency $z$. 
This is followed by Section 3~\ref{sec:discrete_distr} where we introduce the normal-discrete SF model and examine its properties. 
%In this section we demonstrate results of estimation of the commonly used normal-exponential on the data generated by  %the normal discrete model.
 Section 4 concludes the paper. The proofs are provided in Appendix~\ref{sec:proofs}.

\section{Marginal effects of exogenous determinants on technical inefficiency and technical efficiency}
\label{sec:limit_effects}

%Stochastic frontier analysis originates from (\cite{aigner1977formulation, meeusen1977efficiency}). 
%Further improvements of the methods as well as for cross-sectional models as for panel data models are presented e.g. in ~\cite{battese1988prediction,cornwell1990production,heshmati1995efficiency,kumbhakar2000stochastic, greene2005reconsidering,kumbhakar2014technical}.

%The basic SFA model is a parametric production function with a random error, consisting of two components --- stochastic noise and inefficiency term
%(\cite{battese1988prediction,kumbhakar2000stochastic}). 
For cross-sectional data the basic SF model (\cite{aigner1977formulation, meeusen1977efficiency}) is 

\begin{equation}
\label{eq:model_specification}
%	\ln y_i = \beta_0 + f(\ln x_{1i}, ..., \ln x_{ki}; \vecBeta) + v_i - u_i,
	y_i = \beta_0 + f(x_{i}, \vecBeta) + v_i - u_i, i = 1, \ldots, N,
%     = \\
% 	&= \beta_0 + \vecBeta\ln x_{i} + v_i - u_i, \nonumber
\end{equation}
where $y_i$ is log output, $x_{i}$ is a $k \times 1$ vector of inputs (usually in logs), $\vecBeta$ is $k \times 1$ vector of coefficients; $N$ is the number of observations. The production function $f(\cdot)$ usually takes the log-linear (Cobb-Douglas) or the transcendental logarithmic (translog) form. 
The noise and inefficiency terms, $v_i$ and $u_i$, respectively, are assumed to be independent of each other and also independent of $x$. 
The sum $\varepsilon_i = v_i - u_i$ is often labeled as the composed error term.
This assumption is relaxed in some recent papers, see \cite{Lai-Kumhakar2019} and the references therein.

To separate noise from inefficiency the SF models assume distributions for both $v$ and $u$. The popular assumption on the noise term is that $v_i \sim i.i.d. \mathcal{N}(0, \sigv^2)$.
Several alternative assumptions are made on the inefficiency term, $u_i$. The most popular ones are exponential and half-normal.
We refer to these specifications as \emph{the normal-exponential model} and \emph{the normal-half normal model}~ \eqref{eq:model_specification}.

As an alternative we consider a model in which the inefficiency term follows a discrete distribution:
$u$ takes a value $u_1$ with probability $ p$ and a value $u_2$ with probability $1 - p$. 
Here $u_1 > 0, \,  u_2 > 0, \, 0 < p < 1$. 
We refer to this specification as \emph{the normal-discrete model}. 
We show that the behaviour of this model can be richer than the behaviour for the normal-exponential and normal-half normal models.

%In accordance with \citep{kumbhakar2000stochastic,wang1990onestep} ignoring of the heteroscedasticity of the inefficiency term $u_i$ causes biased estimates of the parameters of the frontier function as well as biased estimates of the technical efficiency.
%A popular specification for that heteroscedasticity is $ \ln \sigu = \mathbf{z}' \vecGamma$ 
%(Simar, Lovell, Eekaurt, 1994; Caudill, Ford, Gropper, xxxx). 
%Thus, the standard deviation of the inefficiency error  $\sigu = \exp(\mathbf{z}' \vecGamma) > 0$.

Technical efficiency in model~\eqref{eq:model_specification} can be defined in several ways.
\cite{aigner1977formulation} suggested $\bbE(u)$ as the measure of the mean technical inefficiency. 
Later, Lee and Tyler (1978) proposed $\bbE(e^{-u})$ as the measure of the mean technical efficiency.
Without determinants, these measures are not observation-specific. To make it observation-specific Jondrow et al. (1982) suggested $\bbE(u_i | \varepsilon_i)$ as a predictor of $TI$. 
Following this procedure Battese and Coelli (1988) suggested $\bbE(e^{-u_i}|\varepsilon_i)$ as a predictor of observation-specific measures of $TE$.

Since we model determinants of $TI$ via the $z$ variables in the variance of $u$, $\sigu$, without loss of generality we write $\sigu = \sigu(z)$. For convenience we consider only one $z$ variable.
A popular specification in the literature is $\sigu (z) = \exp(z' \vecGamma) = \exp(\gamma_0 + \gamma z) > 0$.

If $\gamma > 0$, then
\[
\frac{\partial \sigu}{\partial z} = \sigu(z) \, \gamma > 0.
\] 
Thus an increase in $z$ causes  $\sigu$ to increase. 
Intuition tells us that in this case $TI$ measured by either $\bbE (u(z))$ or $\bbE (u(z)|\varepsilon) $ will increase while  $TE$ measured by either $\bbE (e^{-u(z)})$ or $\bbE (e^{-u(z)}|\varepsilon)$ will decrease. 
Below we show that it is true for the normal-exponential and the normal-half normal models. 
However, the situation with the normal-discrete model can be different.

In the next subsections we examine these predictors of $TI$ and $TE$ for the two models: normal-exponential and normal-half normal.
In the next section we move to the normal-discrete model.

\subsection{Exponential distribution of inefficiency}
\label{exponential_1}

The two common models for $u \geq 0$ are an exponential distribution and a half-normal distribution.
If $u$ follows an exponential distribution it has the following probability density function:
\begin{equation}
\label{eq:exp_density}
	f(u) = \frac{1}{\sigu(z)}\exp{\left(-\frac{u(z)}{\sigu(z)}\right)}, \:\: u\geqslant 0,
\end{equation}

%Here and below we omit $i$-th index for short.
%For the normal-exponential model \eqref{eq:model_specification}--\eqref{eq:exp_density} we have \eqref{eq_e_u} and %\eqref{eq_te_exp-u} as 
Technical inefficiency $TI$ and the technical efficiency $TE$ can be predicted from:

\begin{align}
&\bbE(u) = \sigu,
\nonumber \\
&\bbE\left(e^{-u}\right) = \frac{1}{\sigu + 1}. 
\label{eq_te_exp-u}
\end{align}

%%Note that so long as $z$ is observation-specific, the above predictors of TI and TE are observation-specific.

%One can, of course, also use the Jondrow et al. (1982) measure of TI and Battese-Coelli (1988) measure of TE which are

%%One can obtain marginal effect of $z$ on TI and TE from \eqref{eq_e_u}--\eqref{eq_te_exp-u} which are

One can obtain marginal effects of $z$ on 
the mean technical inefficiency 	
TI and 
the mean technical efficiency 	
TE from the equations which are:

\begin{align}
& \frac{\partial \bbE(u)}{\partial z} 
=\frac{\partial \sigu}{\partial z}
 	  			, \\
& \frac{\partial \bbE\left(e^{-u}\right)}{\partial z} =
%\frac{\partial \bbE(e^{-u})}{\partial \sigu}\frac{\partial \sigu}{\partial z}=
-\frac{1}{(\sigu + 1)^2}\frac{\partial \sigu}{\partial z}.
\end{align}
Thus the signs of the marginal effects of $z$ on 
$TI =  \bbE(u)$ and  $TE = \bbE(e^{-u})$
have opposite signs. 
If $z$  increases inefficiency, it will decrease efficiency  and vice versa. 

Instead of using the unconditional mean,
one can use the conditional means \cite{jondrow1982estimation} to estimate $TI$ and the \cite{battese1988prediction} to estimate  $TE$.
These estimators can then be used to compute the marginal effects of $z$.

It is believed that for both the 
unconditional and conditional (observation specific)
estimates of $TI = \bbE(u_i| \varepsilon_i)$ and $TE = \bbE\left(e^{-u_i}| \varepsilon_i\right)$, 
discussed below, the marginal effects of $z$ on $TI$ and $TE$ have opposite signs. 
However, we failed to find a proof of this result in the literature. 
We provide the proof of these results in four Theorems below.

%The most commonly used in the literature definitions are 
%observation-specific measure of inefficiency, $\bbE(u|\varepsilon)$, 
%and observation specific measure of technical efficiency, $\bbE(e^{-u}|\varepsilon)$ (\cite{battese1988prediction}, \cite{kumbhakar2000stochastic}).
In the empirical literature the conditional mean is widely used to estimate both $TI$ and $TE$. 
The advantage of using the conditional means is that the resulting estimates of $TI$ and $TE$ are observation-specific without the $z$ variables explaining inefficiency.  
However, since our focus is the marginal effects, we assume there are determinants.

The conditional mean (\cite{jondrow1982estimation} measure of $TI$ and $TE$ (\cite{battese1992frontier}  (after dropping the `$i$' subscript to avoid clutter of notation) for the normal-exponential case are (\cite{kumbhakar2000stochastic})
\begin{align}
    TI = \bbE(u|\varepsilon) &= \frac{\sigv\phi\left(\frac{\mu_*}{\sigv}\right)}{\Phi\left(\frac{\mu_*}{\sigv}\right)}+\mu_*,
    \label{eq_e_u2} \\
	TE = \bbE(e^{-u}|\varepsilon) &=\frac{\exp\left(-\mu_* + \frac{\sigv^2}{2}\right)
	\Phi\left(\frac{\mu_*}{\sigv} - \sigv\right)}
    {\Phi\left(\frac{\mu_*}{\sigv}\right)}, \label{eq_te_exp} \\
    \mu_* &= -{\varepsilon} - \frac{\sigv^2}{\sigu},  \label{eq_mu_exp}
\end{align}
where ${\varepsilon} = v-u$, $\phi(\cdot)$ is the probability density function and $\Phi(\cdot)$ is the cumulative distribution function of the standard normal variable. In deriving this formula, $v$ is assumed to be $i.i.d.$ normal and $u$ is  $i.i.d.$ exponential (see \cite{kumbhakar2000stochastic}).
%%For observation-specific case we assume that $u_i$ are independent %%given exogenous variables $z$.
Note: both $TI$ and $TE$ are observation-specific.

%The following form of the heteroscedasticity of the inefficiency error term 
%was suggested in \citep{caudill1993biases,caudill1995frontier,Hadri1999estimation,wang2003stochastic}:
%heteroscedasticity depends on a vector of observable variables and the associated parameters.
%\begin{equation}
%	\sigma_{ui}^2 = \exp({z}_{i}'\delta),
%\end{equation}
%here ${z}_{i}$~are heteroscedasticity factors (observable variables), and $\delta$~are parameters.
%To simplify notations we use the following specification with only one factor and intercept:
%\begin{equation}
%\sigma_{ui}^2 = \exp(\delta_0 + \delta z_i), \enspace  \text{or} \enspace \sigma_{ui} = \exp(\gamma_0 + \gamma z_i), 
%\end{equation}
%where $\delta_0 = 2 \gamma_0$ and  $\delta = 2 \gamma$.

The marginal effects of $z$ can be computed from 
 $\frac{\partial \bbE(u|\varepsilon)}{\partial z}$ and $\frac{\partial \bbE(e^{-u}|\varepsilon)}{\partial z}$:
\begin{align}
&\frac{\partial \bbE(u|\varepsilon)}{\partial z} 
= \frac{\partial \bbE(u|\varepsilon)}{\partial \sigu(z)} 
\frac{\partial \sigu(z)}{\partial z}
  			          \label{ne_ineff_eps} \\
&\frac{\partial \bbE(e^{-u}|\varepsilon)}{\partial z} =
\frac{\partial \bbE(e^{-u}|\varepsilon)}{\partial \sigu(z)}
\frac{\partial \sigu(z)}{\partial z}
                                        \label{ne_eff_eps}
\end{align}

So, to prove that marginal effects of $z$ on of the inefficiency and the technical efficiency have opposite 
signs, it is enough to prove that the marginal effects of
%$\sigu(z)$ 
$\sigu$
on $TI$ and $TE$ have 
opposite signs
\footnote{In some papers (e.g.~\cite{ruggiero1999efficiency}; 
	\cite{ondrich2001efficiency}) efficiency is defined as $E(-u| \varepsilon)$, thus, these marginal effects are opposite by definition.}. 

We derive these in Statements~\ref{statement:e}, \ref{statement:te} and prove the result about signs
in Theorems~\ref{theorem:sign_exponential_e} and~\ref{theorem:sign_exponential}.
To avoid notational clutter from now on we write $\sigu$ instead of $\sigu(z)$.

\begin{statement}
\label{statement:e}
For the normal-exponential model \eqref{eq:model_specification}--\eqref{eq:exp_density}
the marginal effect of the $\sigu$ on the inefficiency~\eqref{eq_e_u2} is:
\begin{align}
&\frac{\partial \bbE(u|\varepsilon)}{\partial\sigu} =
 \frac{\sigv^2}{\sigu^2}
 \left(\frac{\Phi^2(t)-\phi^2(t)-t\phi(t)\Phi(t)}{\Phi^2(t)}\right),
\label{eq_marg_exp_e}
\end{align}
where \, $t=\frac{\mu*}{\sigv} =
 -\frac{\varepsilon}{\sigv} - \frac{\sigv}{\sigu}$.
\end{statement}

%%Proof of the Statement~\ref{statement:e}}
\begin{proof}
\begin{align}\nonumber
\frac{\partial \bbE(u|\varepsilon)}{\partial\sigu}
&= \frac{\partial \bbE (u|\varepsilon)}{\partial t}\frac{\partial t} {\partial\sigu} = \frac{\sigv}{\sigu^2} 
\frac{\partial}{\partial t}\left(\sigv\frac{\phi(t)}{\Phi(t)}+ z \sigv \right) = \\
&= \frac{\sigv^2}{\sigu^2}
\left(\frac{\Phi^2(t)-\phi^2(t)-t\phi(t)\Phi(t)}{\Phi^2(t)}\right). \nonumber
\end{align}
%%\qed
\end{proof}

\begin{statement}
\label{statement:te}
For the normal-exponential model \eqref{eq:model_specification}--\eqref{eq:exp_density}
the marginal effect of the $\sigu$  on technical efficiency $TE = \bbE (\exp(-u)|\varepsilon)$ equals:

\begin{align}
	\frac{\partial TE}{\partial\sigu} &= \frac{\sigv}{\sigu^2}\cdot
	\frac{\exp\left(-t \sigv + \frac{\sigv^2}{2}\right)}
	{\Phi^2(t)} 
	\times  \nonumber  \\
    &\times \bigl(-\sigv\Phi(t-\sigv)\Phi(t) 
	+ \phi(t-\sigv)\Phi(t) - \Phi(t-\sigv)\phi(t)\bigr),
	\label{eq_marg_exp}
\end{align}
where as before \, $t=\frac{\mu*}{\sigv} =
-\frac{\varepsilon}{\sigv} - \frac{\sigv}{\sigu}$.
\end{statement}

%%Proof of the Statement~\ref{statement:te}}
\begin{proof}
 From~\eqref{eq_te_exp}--\eqref{eq_mu_exp} we get:
\begin{align}\nonumber
TE &= \bbE(e^{-u}|\varepsilon) 
    =\frac{\exp\left(-t\sigv + \frac{\sigv^2}{2}\right) 	\Phi(t - \sigv)}
     {\Phi\left(t\right)},   \nonumber
\end{align}
thus
\begin{align}\nonumber
\frac{\partial TE}{\partial\sigu}  
  &= \frac{\partial TE}{\partial t}  \frac{\partial t}{\partial \sigu} 
   = \frac{\sigv}{\sigu^2} \frac{\partial }{\partial t} 
   \frac{\exp\left(-t\sigv + \frac{\sigv^2}{2}\right)\Phi(t - \sigv)} {\Phi(t)} \\
  &= \frac{\sigv}{\sigu^2}\cdot
   \frac{\exp\left(-t \sigv + \frac{\sigv^2}{2}\right)}
   {\Phi^2(t)} 
   \times \\ \nonumber
   &\times \bigl(-\sigv\Phi(t-\sigv)\Phi(t) 
   + \phi(z-\sigv)\Phi(t) - \Phi(t-\sigv)\phi(t)\bigr).  \nonumber
\end{align}

%%	\qed

\end{proof}

\begin{theorem}
\label{theorem:sign_exponential_e}
For the normal-exponential model defined by \eqref{eq:model_specification} and \eqref{eq:exp_density}
the marginal effect of  $\sigu$ on $E(u|\varepsilon)$ is non-negative. 
That is, if  $\sigu$ increases, technical inefficiency estimated by $E(u|\varepsilon)$ also increases:
\[
	% \label{eq_dif0_u}
		\frac{\partial E(u|\varepsilon)}{\partial \sigu}\geqslant 0
\]
\end{theorem}

\begin{theorem}
\label{theorem:sign_exponential}
For the normal-exponential model defined by \eqref{eq:model_specification} and \eqref{eq:exp_density}
the marginal effect of  $\sigu$ on  $TE=E(e^{-u}|\varepsilon)$ is non-positive. 
That is, if  $\sigu$ increases, $TE$ decreases:
\[
	% \label{eq_dif0_expu}
		\frac{\partial \bbE( e^{-u}|\varepsilon)}{\partial \sigu}\leqslant 0.
\]
\end{theorem}

Proofs of 
Theorems~\ref{theorem:sign_exponential_e}-\ref{theorem:sign_exponential} 
are given in Appendix~\ref{sec:proofs}.

\subsection{Half-normal distribution of inefficiency}
\label{half_normal_2}

If $u$ follows a half-normal distribution it has the following probability density function:
\begin{equation}
\label{eq:half_normal_density}
	f(u) = \frac{\sqrt{2}}{\sqrt{\pi} \sigu(z)}\exp{\left(-\frac{u(z)^2}{2 \sigu^2(z)}\right)}, \:\: u\geqslant 0,
\end{equation}

The technical inefficiency $TI$ and the technical efficiency $TE$ can be measured as (see, e.g. \cite{kumbhakar2000stochastic}):

\begin{align}
&\bbE(u) = \sigu\sqrt{\frac{2}{\pi}},                  
\nonumber \\
&\bbE\left(e^{-u}\right) = 2\left(1-\Phi(\sigu)\right)\exp\left(\frac{\sigu^2}{2}\right). 
\label{eq_te_hl-u}
\end{align}

One can obtain marginal effects of $z$ on 
the mean technical inefficiency 	
TI and 
the mean technical efficiency 	
TE from the equations which are:

\begin{align}
& \frac{\partial \bbE(u)}{\partial z} 
=\sqrt{\frac{2}{\pi}}\frac{\partial \sigu}{\partial z}, \label{eq:marginal_ti} \\
& \frac{\partial \bbE\left(e^{-u}\right)}{\partial z} =
2\frac{\partial\sigu}{\partial z}\exp\left(\frac{\sigu^2}{2}\right) \left(\sigma_u - \phi(\sigu) - \Phi(\sigu)\sigu\right).
\label{eq:marginal_te}
\end{align}

\begin{statement}
Marginal effects on TI and TE in~\eqref{eq:marginal_te} and~\eqref{eq:marginal_ti}
have different signs.
\end{statement}

The statement follows from the negativity of $x - \phi(x) - x \Phi(x)$,
for example, see inequality (2) in~\cite{sampford1953some}: $\frac{\phi(x)}{1 - \Phi(x)} > x$.

The conditional mean measure of $TI$ \citep{jondrow1982estimation} and $TE$ \citep{battese1992frontier} for the normal-half normal case are \citep{kumbhakar2000stochastic}

\begin{align}
    \bbE(u|\varepsilon) &= \frac{\sigs\phi\left(\frac{\mu_*}{\sigs}\right)}{\Phi\left(\frac{\mu_*}{\sigs}\right)}+\mu_*,
    \label{eq_e_u2_hn} \\
	TE = \bbE(e^{-u}|\varepsilon) &=\frac{\exp\left(-\mu_* + \frac{\sigs^2}{2}\right)
	\Phi\left(\frac{\mu_*}{\sigs} - \sigs\right)}
    {\Phi\left(\frac{\mu_*}{\sigs}\right)}, \label{eq_te_hn} \\
    \mu_* &= \frac{-\sigu^2\varepsilon}{\sigv^2+\sigu^2}, \label{eq_mu_hn}\\
    \sigs^2 &= \frac{\sigv^2\sigu^2}{\sigv^2+\sigu^2}.\label{eq_sgm_hn}
\end{align}

\begin{theorem}
\label{theorem:sign_half_normal_e}
For the normal-half normal model defined by \eqref{eq:model_specification} and \eqref{eq:half_normal_density}
the marginal effect of  $\sigu$ on $E(u|\varepsilon)$ is non-negative. 
That is, if  $\sigu$ increases, technical inefficiency estimated by $E(u|\varepsilon)$ also increases:
\[
	% \label{eq_dif0_u}
		\frac{\partial E(u|\varepsilon)}{\partial \sigu} \geqslant 0
\]
\end{theorem}

\begin{theorem}
\label{theorem:sign_half_normal}
For the normal-half normal model defined by \eqref{eq:model_specification} and \eqref{eq:half_normal_density}
the marginal effect of  $\sigu$ on  $TE=E(e^{-u}|\varepsilon)$ is non-positive. 
That is, if  $\sigu$ increases, $TE$ decreases:
\[
	% \label{eq_dif0_expu}
		\frac{\partial \bbE( e^{-u}|\varepsilon)}{\partial \sigu}\leqslant 0.
\]
\end{theorem}

Thus taking into account \eqref{ne_ineff_eps}, \eqref{ne_eff_eps} and 
Theorems \ref{theorem:sign_exponential_e}-\ref{theorem:sign_half_normal} we conclude that for the normal-exponential model \eqref{eq:model_specification}, \eqref{eq:exp_density} as well as for the normal-half normal model \eqref{eq:model_specification}, \eqref{eq:half_normal_density} 
signs of marginal effects of  $z$ on $E(u|\varepsilon)$ and $TE=E(e^{-u}|\varepsilon)$  are opposite, i.e.,
\begin{equation}
\label{marginals_exponential}
\mathrm{sign} \frac{\partial \bbE(u|\varepsilon)}{\partial z} 
  = -\mathrm{sign} \frac{\partial \bbE(e^{-u}|\varepsilon)}{\partial z} .
  \end{equation}

Proofs of Theorems~\ref{theorem:sign_half_normal_e}-\ref{theorem:sign_half_normal} 
are given in Appendix~\ref{sec:proofs}.

\section{Discrete distribution of inefficiency error}
\label{sec:discrete_distr}

\subsection{Discrete model}

To come up with a counter-example of the above result, we now consider an example of a discrete distribution for $u > 0$ with the support that consists of two values $u_1$ and $u_2$: 

\begin{equation}
u = 
\begin{cases}
u_1, \,\, \text{with } P(u = u_1) =  p,\\
u_2, \,\, \text{with } P(u = u_2) =  1 - p,\\
\end{cases}
     \label{eq:discr2}
\end{equation}
with $u_1 > 0,  \, u_2 > 0,  \,  0 < p < 1$.

For the distribution of $u$ in \eqref{eq:discr2} we have:
\begin{align}
\bbE(u) &= u_1p+u_2(1-p), \nonumber \\
\bbV(u) &= \sigu^2 = p (1 - p) (u_1 - u_2)^2, \label{eq_dis_var}  \\
TE(u) &= \bbE(e^{-u}) = pe^{-u_1}+(1-p)e^{-u_2}.\label{eq_dis_te}
\end{align}

The proposed normal-discrete model is identifiable model, as our study in Appendix~\ref{sec:identifiability_check} shows.

In contrast to the exponential distribution \eqref{eq:exp_density} standard deviation $\sigu$ of this distribution depends on three parameters $u_1$, $u_2$, and $p$.

%%%Thus, it can increase after changing of these parameters. 

\subsection{Numerical experiments}

Use of this discrete distribution can result in unexpected behavior of $TI$ and $TE$ with an increase in $\sigu$ induced by an increase in $z$.

To show this we consider an example with the factor variable $z$, such that $9 \le z \le 17$ and
\begin{equation}
\label{example_discr}
\begin{cases}
     p &= 0.9 + 0.001 z, \\
     u_1 &= 0.1, \\ 
     u_2 &= 1 + 0.2 z.   
\end{cases}
\end{equation}
so that $\sigu(z)$ is an increasing function of $z$ (left pane of Fig. 1). But in the range  $10.5 \le z \le 17$ the behavior of $TI$ and $TE$ are "abnormal", see the right pane of Fig. 1. %Figures~\ref{fig:discrete_sigu_vs_z} and %\ref{fig:discrete_te_mean_u_vs_z}. 
In this range both $TI$ and $TE$ are increasing functions of $\sigu$. 
The variance $\sigu$ is an increasing function of $z$. 
That is, an increase in $z$ causes an increase of $\sigu$ which causes a simultaneous increase of $TI$ and $TE$.

\begin{figure}
    \centering
    \begin{subfigure}[b]{0.42\textwidth}
        \includegraphics[width=\textwidth]{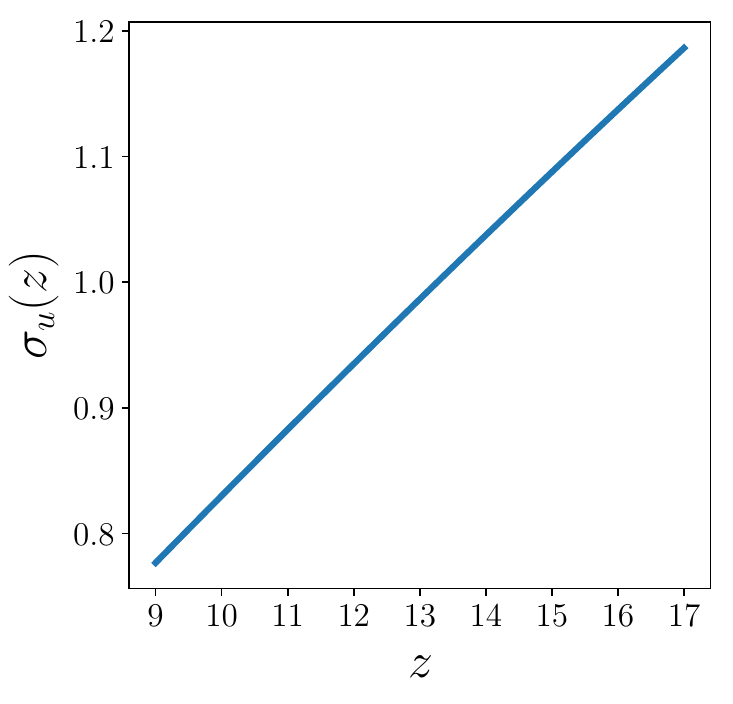}
        \caption{Variance $\sigu$ is an increasing function of $z$ for the considered normal-discrete model}
        \label{fig:discrete_sigu_vs_z}
    \end{subfigure}
    ~ %add desired spacing between images, e. g. ~, \quad, \qquad, \hfill etc. 
      %(or a blank line to force the subfigure onto a new line)
    \begin{subfigure}[b]{0.485\textwidth}
        \includegraphics[width=\textwidth]{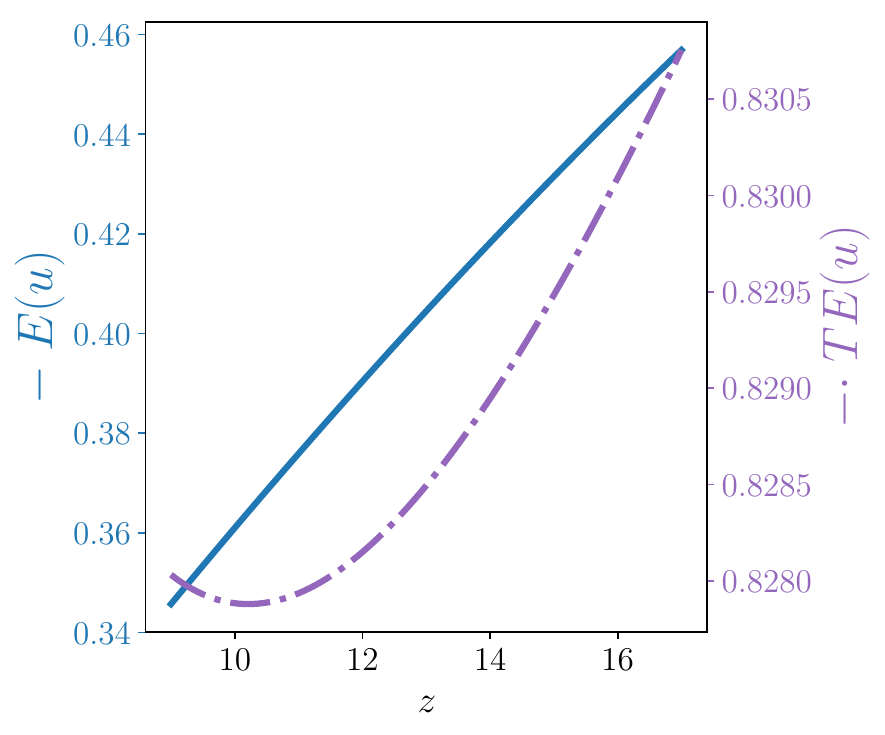}
        \caption{$\bbE(u)$ and $TE = \bbE (e^{-u})$ 
        	in the range  $10.5 \le z \le 17$ 
        	are both monotonically increasing function of $z$ and thus of~$\sigu$}
        \label{fig:discrete_te_mean_u_vs_z}
    \end{subfigure}
    \caption{Unusual behaviour of the discrete normal model}\label{fig:unusual_normal_discrete}
\end{figure}

%\begin{figure}[h]
%	\centering
%	\includegraphics[width=0.7\linewidth]{book01}
%	\caption{$E(u)$ (left axis) and $TE=E(exp(-u))$ (right axis) as functions of $z$}
%	\label{fig:book01}
%\end{figure}

Thus if in reality the distribution of $u$ is discrete as in \eqref{example_discr}, that is, $u$ is generated from the discrete distribution and $v$ is generated from a normal distribution so that the model is a normal-discrete model specified by \eqref{eq:model_specification} and \eqref{eq:discr2}, and one applies the normal-exponential model \eqref{eq:model_specification} and \eqref{eq:exp_density},
 the estimates are likely to suffer from model misspecification.
Use of the normal-exponential model according to
\eqref{ne_ineff_eps}, \eqref{ne_eff_eps} an increase in $z$ causes a decrease of $TI$, while the real situation is the opposite.

\subsection{Discrete distribution. Mean $TE$}
\label{discrete_1}
	
To illustrate the aforementioned problem we run simulations with the following specifications.
We choose the sample size $N = 1000$. 
The single input $x_i$ is generated from an uniform distribution defined for the interval $[2, 7]$. 
The noise term $v_i \sim N(0, 0.25)$. A single variable $z_i$ comes from an uniform distribution defined in the interval $[9, 17]$. 
The parameters of the discrete distribution of $u$ in \eqref{eq:discr2} are: $u_{i, 1} = 0.1$; $u_{i, 2} = 1 + 0.2 z_i$; $p_i = 0.9 + 0.001 z_i$. 
To simulate $u_i$, we also define an uniformly distributed random variable $r_i \sim U[0, 1]$ for each $i$. We then
assign $u_i = u_{i, 1}$ if $r_i < p_i$ and $u_i = u_{i, 2}$ otherwise. 
Finally we generate output $y_i$ according to $y_i = 1 + x_i + v_i - u_i$.

Using the generated data we estimated parameters of normal-exponential model \eqref{eq:model_specification} and \eqref{eq:exp_density}
with the following specification for $\sigu(z)$, viz., $\ln \sigu (z_i)= \gamma_0 + \gamma z_i$,
%% we
and 
obtained
\[
\hat \sigma_{u_i} = \exp(-0.618 + 0.025 z_i).
% \hat \sigma_{u_i} = exp(-.6188313+ .0259041 z_i).
\]
We used this estimate of $\sigu(z_i)$ to get estimate of $TE$ using \eqref{eq_te_exp-u}, i.e.,
$\widehat {TE}_i = {1}/({1 + \hat \sigma_{ui}})$.

Plot of  true $\sigma_{u_i}$ calculated using \eqref{eq_dis_var} and estimated $\hat \sigma_{u_i}$ on $z$ is presented in Figure~\ref{fig:Graph_sigma}. Similarly, plot of true $TE_i$ calculate using \eqref{eq_dis_te} and estimated $\widehat {TE}_i$ on $z$ is presented in Figure \ref{fig:Graph_te}. It can be seen from the figures that while $\hat \sigu$ increases with $z$, like $\sigu$,  true $TE$  and the estimate of $TE$ move in opposite directions. In this case 
the model misspecification leads to the wrong conclusion of the negative effect of $z$ on $TE$.
 
\begin{figure}[h]
	\centering
	\includegraphics[width=0.6\linewidth]{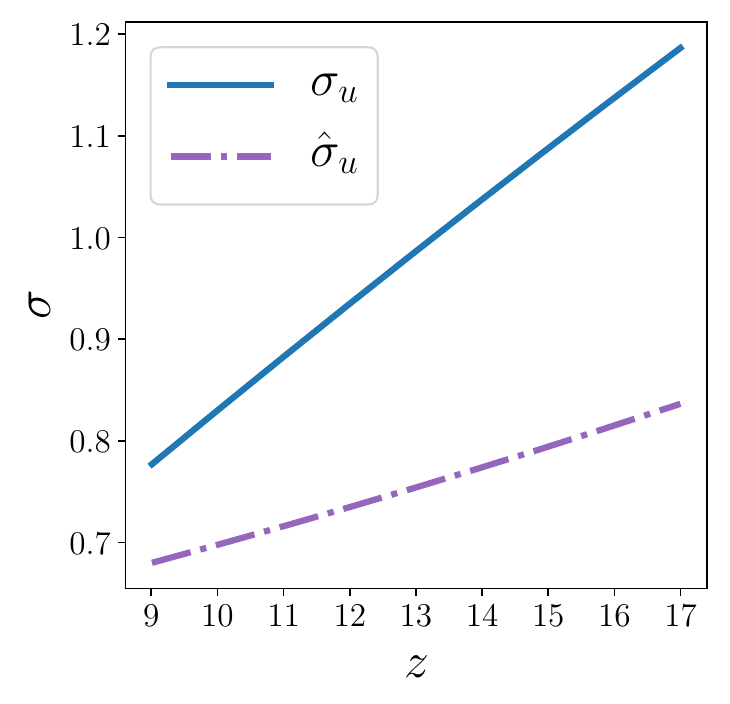}
	\caption{$\sigu$ and $\hat \sigu$ behave in a similar way for the normal-discrete model}
	\label{fig:Graph_sigma}
\end{figure}
\begin{figure}[h]
	\centering
	\includegraphics[width=0.6\linewidth]{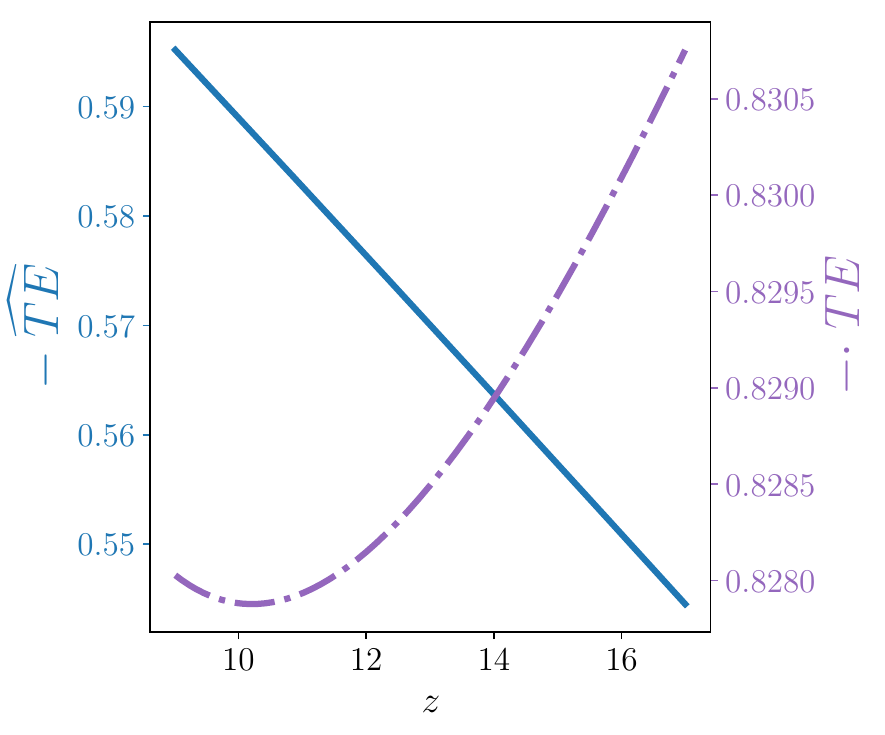}
	\caption{$TE$ and $\widehat {TE}$ as function of $z$ behave in a different way for the normal-discrete model}	
	\label{fig:Graph_te}
\end{figure}

%%\subsection{ $TE$ estimated as conditional mean}
\subsection{Discrete distribution. Observation-specific $TE$}
\label{discrete_2}

%In applications it is important to identify the dependence of observation specific technical efficiency from the heteroscedasticity factor. 
%For the normal-exponential model Theorems~\ref{theorem:sign_exponential_e}, \ref{theorem:sign_exponential}
%suggest that the inefficiency and the technical efficiency  move in 
%opposite
%%%different 
%directions, when a heteroscedasticity factor $z$ increases. 

We continue with the discrete case to provide another counter-example when $TE$ is estimated from the conditional mean.
%estimate examine in a similar way the normal-discrete model 
%\eqref{eq:model_specification} and \eqref{eq:discr2}.
%There exists a counterexample with the behaviour that is different from that of the normal-exponential model.
%So, if one applies the normal-exponential model in this case, then due to the model misspecification the marginal effect of factor $z$ is wron g at least for some observations.
For this we consider a discrete random variable $u > 0$ 
which takes values $u_i = z \, u_{i0}$, $i=1, 2$ with probabilities $p_1, p_2$, such that $p_1 + p_2 = 1$, and $u_{i0} > 0$,  $i = 1, 2$, 
 $z>0$.

\begin{equation}
\label{u_random_discr3}
P(u_i = z u_{i0}) = p_i,  \,\,\, i = 1, 2.
\end{equation}

Variance of  $u_i$ depends on  $z$, i.e., 

\begin{align}
\sigma^2_{ui} &= z^2 p_1 p_2 (u_{10} - u_{20})^2 = z^2 c^2,  \,\,\,\, c > 0, 	\label{variance_discr3}
\end{align}
where $c = p_1 p_2 (u_{10} - u_{20})$.
Thus 
\begin{equation}
\label{std_discr3}
\sigu = z\,c,  \,\, \text{and}   \,\, 
\frac{\partial \sigu}{\partial z} = c > 0,
\end{equation}

\begin{statement}
\label{statement:effect_te_sign}
Consider the SF model \eqref{eq:model_specification} with  $v_i\sim \mathcal{N}(0, \sigv^2)$ and a one-parameter distribution for $u$ in \eqref{u_random_discr3}.
Then the sign of the marginal effect of  $z$ 
on $TE$ defined as $TE=E(e^{-u}|\varepsilon)$ 
is:
\begin{align*}
\label{marginal_TE_discr3}
\frac{\partial TE}{\partial z} 
&= \frac{\partial \bbE (e^{-u}|\varepsilon)}{\partial z}  \\
&=-\frac{1}{\sum\nolimits_{i=1}^{2}p_i e^{-w_i}}
\sum\nolimits_{i=1}^{2}p_i e^{-z\, u_{i0}}e^{-w_i} (u_{i0}+w'_i) \\
&+ \frac{1}{ \left( \sum\nolimits_{i=1}^{2}p_i e^{-w_i} \right)^2 }
\left(  \sum\nolimits_{i=1}^{2}p_i e^{-z\, u_{i0}}e^{-w_i} \right)
\left( \sum\nolimits_{i=1}^{2}p_i e^{-w_i}w'_i \right),
\end{align*}
where $w_i=\frac{(z\, u_{i0}+\varepsilon)^2}{2\sigv^2}$ and 
$w'_i=\frac{\partial}{\partial z}w_i =\frac{z\, u_{i0}^2 +\varepsilon\,u_{i0}}{\sigv^2}$.
\end{statement}
The proof is presented in the Appendix.

Note that the marginal effect of $z$ on $TE$ 
in the normal-exponential model is negative if  
$\frac{\partial \sigu}{\partial z} > 0$ (see Theorem~\ref{theorem:sign_exponential}). However, in the normal-discrete model, the sign of the marginal effect of $z$ depends on value of $\varepsilon$. That is, the value of the marginal effect as well its sign depends on the value of $\varepsilon$. 

We illustrate this with the plot of
$\frac{\partial TE}{\partial z} $ against $\varepsilon$ for these values of the model 
parameters:  $\, z=8.5;\, \sigv=1;\, u_1=0.1; \,
	u_2=0.89;\, p_1=0.99;\, p_2=0.01$

\begin{figure}[H]
	\centering
	\includegraphics[width=0.6\linewidth]{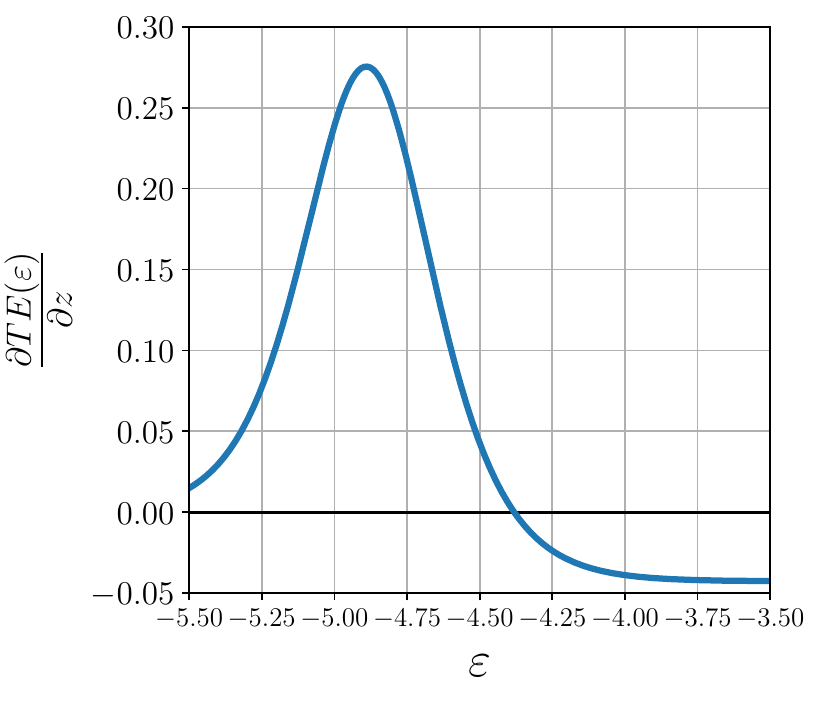}
	\caption {Marginal effect $\frac{\partial TE}{\partial z}$  
		as function of $\varepsilon$
	for $z=8.5;\, \sigv=1;\, u_1=0.1; \,
	u_2=0.89;\, p_1=0.99;\, p_2=0.01$}	
    %  ==== SEE scalar2.do ====
	\label{fig:scalar2}
\end{figure} 

From Figure~\ref{fig:scalar2} one can see that if the normal-discrete model is the true model, then the sign of the marginal effect may vary across observations.  
% {\color{red} How reasonable are the values of $\varepsilon$ in the figure?} 
% {\color{blue} -- what "is reasonable" ?? skip this}
But 
for the normal-exponential model the marginal effect 
 is always negative if  $\frac{\partial \sigu}{\partial z} > 0$.
Thus if the normal-exponential model is used, where the true model is normal-discrete, one can come to the wrong conclusion regarding the sign of the marginal effect.

Sometimes the focus is not on the individual values of $TE$ but their rankings. To examine how the true values of $TE$ are related to their estimated counterparts for the simulated model, we consider the following simulations.
We used $N=1000$, generated input $x_i$ from a uniformly distributed random variable in the interval $[2, 2.3]$.  The noise term 
is 
generated from $v_i \sim N(0,\,1)$. 
The $z_i$ variable 
is 
generated from a uniformly distributed random variable in the interval $[8, 9.4]$. 
The parameters of the distribution of the discrete distribution of $u$ are chosen as: $p_1=0.8, p_2=0.2; u_{(1)}=0.1, u_{(2)}=0.89$. We also generated a variable $r_i$ which is 
uniformly distributed in the interval $[0, 1]$. Then we generated $u_{i0}=u_{(2)}$ if $r_i<p_2$ and $u_{i0}=u_{(1)}$ otherwise, and assume $u_i = z_i \, u_{i0}$. 
Finally we generated output $y_i$ as:
 $y_i = 1 + x_i + v_i - u_i$.

We used these data to estimate the parameters of the normal-exponential model \eqref{eq:model_specification}--\eqref{eq:exp_density}
with the following specification of $\sigu$:
 $\ln \sigma_{ui} = \gamma_0 + \gamma z_i$,
and obtained 
%%the estimates of TE ($\widehat{TE}_i$)  from the formula
the estimates of the observation specific technical efficiencies 
$\widehat{TE}_i$. 

For each $i$ true $\widehat{TE}_i$ was calculated as
\begin{equation}
TE_i = \bbE(e^{-u}|\varepsilon_i)
     = \frac{ \left( \sum_{i=1}^{k}p_i e^{-z_i\, u_i}e^{-w_i} \right) }
       {\left( \sum_{i=1}^{k}p_i e^{-w_i} \right) },
\end{equation}
where $w_i =\frac{(z_i\, u_i + \varepsilon_i)^2} {2\sigv^2}$.

A scatter plot of the estimated TE,  $\widehat{TE}_i$ against true $TE_i$
provided in Figure \ref{fig:hat_TE_against_TE}. 
It can be seen that in case of positive true marginal effect $TE$ we get confusing values as estimates,
while modeling capability of normal-exponential model is better if the signs of marginal effects coincide.

\begin{figure}
    \centering
    \begin{subfigure}[b]{0.485\textwidth}
        \includegraphics[width=\textwidth]{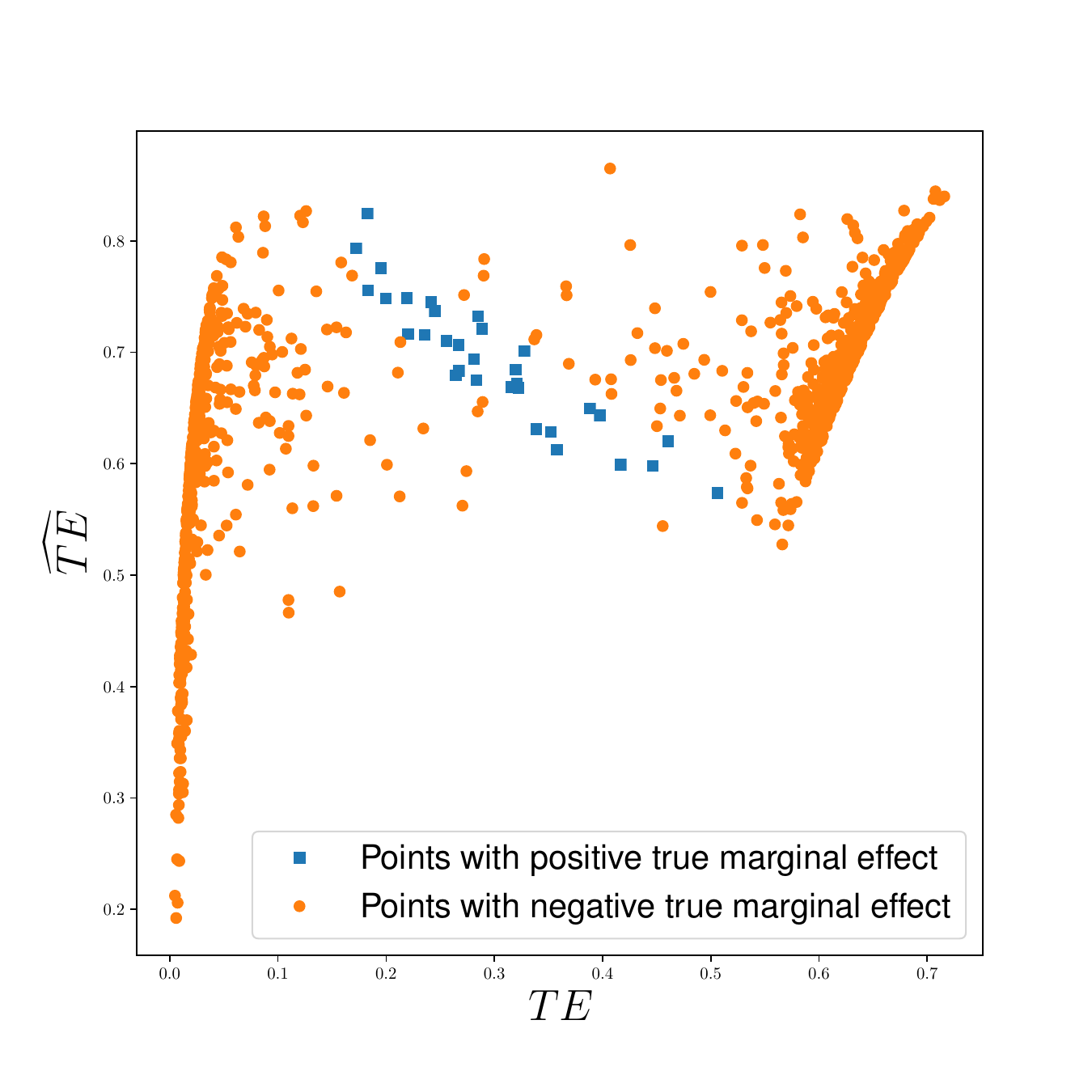}
        \caption{Comparison of all estimates of $\widehat{TE}$ and true $TE$}
    \end{subfigure}
    ~ %add desired spacing between images, e. g. ~, \quad, \qquad, \hfill etc. 
      %(or a blank line to force the subfigure onto a new line)
    \begin{subfigure}[b]{0.485\textwidth}
        \includegraphics[width=\textwidth]{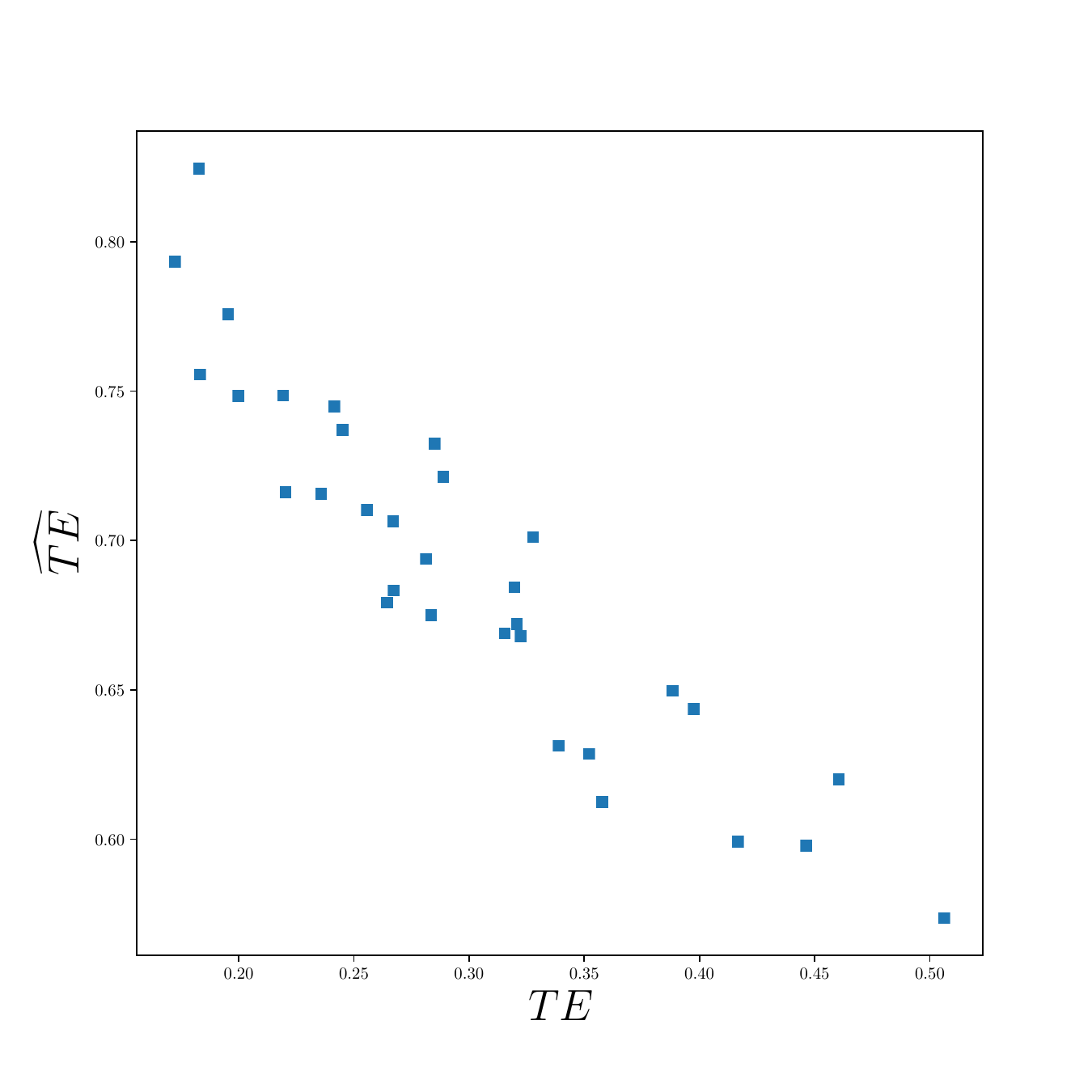}
        \caption{Selected points with $\varepsilon$ between $-2.3$ and $-2.1$}
    \end{subfigure}
    \caption{Scatter plot of $\widehat{TE}$ and true $TE$}	
    \label{fig:hat_TE_against_TE}
\end{figure}

\section{Conclusions and discussions}
\label{sec:conlusions}
%% TODO ====== revision is required =========

%The normal-exponential SFA model with the normal distribution of random shocks and the exponential distribution of the inefficiency error term is widely used in applied research. 
%Incorporating in the model dependence of the standard deviation of the inefficiency error make it possible to make conclusions of the influence of various factors on the efficiency and inefficiency of production units (\cite{wang2002dependence}).
In this paper we derived the formula for computing the marginal effects of determinants of inefficiency ($z$) on both the unconditional and conditional means of technical inefficiency and efficiency for the normal-exponential stochastic frontier model. We proved that for the normal-exponential model the signs of the marginal effects of $z$ on the technical inefficiency and technical efficiency are of opposite sign.

We considered an example of discrete distribution for technical inefficiency and showed that the relationship between the true and estimated technical efficiency for the normal-discrete model can be substantially different from the normal-exponential model, at least for some values of $z$.
%The provided example considers the family of discrete probability distributions.
%However, it is easy to see that arbitrary close continuous probability distribution with similar properties exist.
%For example, we can use a mixture of Gaussian distributions with means $u_1, \ldots, u_k$
%and sufficiently small variances.
%For these models signs of the marginal effects of the heteroscedasticity factor on inefficiency and efficiency can coincide. 
%It is the case for both the mean and observation-specific marginal effects. 
%Moreover, the direction of these marginal effects may vary over observations, unlike the normal-exponential model. 
This results illustrates that if the real world data on noise and inefficiency comes from a  normal and a discrete distribution and a researcher estimates the model assuming that the errors are normal and exponential instead, results on estimated efficiency, its marginal effect and rankings might be all wrong. That is, the consequence of misspecification of inefficiency distribution can be quite serious.
%due to the the misspecification, she can come to the wrong conclusions on the signs of the marginal effects of heteroscedasticity factors on mean or observation-specific technical efficiency or inefficiency.
%The ranking of the production units by their estimated technical efficiency can be as well different from the ranking by the true technical efficiency. 
%So, the SFA models are limited, as they can't handle correctly all possible signs of marginal effects that can happen in real life.

\section*{Acknowledgments}
We are grateful for the invaluable comments provided by participants at the 
Sixteenth European Workshop on Efficiency and Productivity Analysis in London, 2019.

\appendix

\section{Proofs} 
\label{sec:proofs}

\subsection{Proof of Theorems~\ref{theorem:sign_exponential_e} and~\ref{theorem:sign_exponential}}

%%\subsection{Proof of the Theorem~\ref{theorem:sign_exponential_e}}
First we reproduce a proof of the Lemma~\ref{Lemma0} from~(\cite{sampford1953some}):
\begin{lemma}\label{Lemma0}
Let $\phi(z)$ and $\Phi(z)$ be the probability density function and the cumulative density function of the standard normal distribution $\mathcal{N}(0, 1)$,
and $\lambda(z) = \frac{\phi(z)}{\Phi(z)}$. 
Then it holds:
\begin{enumerate}
\item $1 - z\lambda(z) - \lambda(z)^2 \ge 0$
\item $\lambda(z)$ is a decreasing function and its derivative 
$\lambda'(z) \in (-1, 0)$.
\end{enumerate}
\end{lemma}
\begin{proof}
Obviously $f(t) = \frac{\phi(t)}{\Phi(z)}= \frac{\phi(t)}{P(Z \le z)}$ is a probability density function of a random variable $X$ defined at the interval $(-\infty, z)$. 
\begin{align*}
\bbE(X) &= \int\limits_{-\infty}^z t\frac{\phi(t)}{\Phi(z)}dt 
     = \frac{1}{\Phi(z)} \int\limits_{-\infty}^z t\phi(t)dt
     = -\frac{1}{\Phi(z)} \int\limits_{-\infty}^z \phi'(t)dt
     = -\frac{\phi(z)}{\Phi(z)}
     =- \lambda(z), \\
\bbE(X^2) &=   \int\limits_{-\infty}^z t^2\frac{\phi(t)}{\Phi(z)}dt
     = \frac{1}{\Phi(z)} \int\limits_{-\infty}^z t^2\phi(t)dt
     = -\frac{1}{\Phi(z)} \int\limits_{-\infty}^z t\phi'(t)dt = \\
     &= -\frac{1}{\Phi(z)} 
       \biggl( \Bigl. t\phi(t) \Bigr|_{-\infty}^z -  \int\limits_{-\infty}^z \phi(t)dt \biggr)
     = -\frac{1}{\Phi(z)} \bigl(  z\phi(z) - \Phi(z) \bigr)
     = 1-z\lambda(z).
\end{align*}
Hence, the variance is
\[
\bbV(X) = 1-z\lambda(z)-(-\lambda(z))^2 = 1 - z\lambda(z) - \lambda(z)^2 \ge 0.
\]
Since 
\begin{align*}
\lambda'(z) &= \left( \frac{\phi(z)}{\Phi(z)} \right)'  
  = \frac{1}{\Phi(z)^2} \bigl(\phi(z)'\Phi(z) - \phi(z)\Phi(z)'\bigr) = -z\lambda(z)-\lambda(z)^2\\
  & = \bbV(X)- 1,
\end{align*}   
we have $-1 \le \lambda'(z) \le 0 $.
\end{proof}

\subsubsection{Proof of Theorem~\ref{theorem:sign_exponential_e}}

\begin{proof}
From Statement~\ref{statement:e} we have 
\begin{align}
\frac{\partial \bbE(u|\varepsilon)}{\partial\sigu}
 &= \frac{\sigv^2}{\sigu^2}\frac{\Phi^2(z)-\phi^2(z)-z\phi(z)\Phi(z)}{\Phi^2(z)}
 = \frac{\sigv^2}{\sigu^2} \bigl( 1 -z\lambda(z) -z\lambda(z)^2 \bigr), \nonumber
 \end{align}
which is non-negative by Lemma~\ref{Lemma0}.
\end{proof}

\subsubsection{Proof of Theorem~\ref{theorem:sign_exponential}}

\begin{proof}
	From Statement~\ref{statement:te} we have 
\begin{align}
\frac{\partial TE}{\partial\sigu} &= \frac{\sigv}{\sigu^2}\cdot
	\frac{\exp\left(-z \sigv + \frac{\sigv^2}{2}\right)}	{\Phi^2(z)} 
     \Phi(z)\Phi(z-\sigv)\bigl(-\sigv 
	+ \lambda(z-\sigv) - \lambda(z)\bigr).  \label{th2_pr1}
\end{align}
Since the first factors in \eqref{th2_pr1} and $\sigv$ are greater or equal to $0$, it is enough to prove that
\[
 f(t)= -t+ \lambda(z - t) - \lambda(z) \le 0 \text{\,  for all  \,} t \ge 0. 
\]
We have $f(0) = 0$, and $f'(t) = -1-\lambda'(z-t) \le 0$ since $-1 \le \lambda(t) \le 0$ for all $t$ (Lemma~\ref{Lemma0}).
Thus $f(t) \le 0$, and Theorem~\ref{theorem:sign_exponential} is proven.
\end{proof}

\subsection{Proof of Theorems~\ref{theorem:sign_half_normal_e} and~\ref{theorem:sign_half_normal}}

\begin{statement}
\label{statement:4}
For $\lambda(z) = \frac{\phi(z)}{\Phi(z)}$ it holds that:
\begin{equation}
\label{eq:lambdaz}
2 \lambda^2(z) > 1 - z^2 - 3 z \lambda(z)  \text{ for } z < 0.
\end{equation}
\end{statement}

\begin{proof}
According to the proof of Theorem 9 in (\cite{gasull2014approximating}) we get that:
\[
2 + x^2 a^2(x) - a^2(x) - 3 x a(x) > 0 \text{ for } x > 0,
\]
where
\[
a(x) = \frac{1 - \Phi(x)}{\phi(x)} = \frac{1}{\lambda(-x)}.
\]

So,
\[
2 + x^2 \frac{1}{\lambda^2(-x)} - \frac{1}{\lambda^2(-x)} - 3 x \frac{1}{\lambda(-x)} > 0 \text{ for } x > 0. 
\]
By the change of variable $z = -x$ we get:
\[
2 + z^2 \frac{1}{\lambda^2(z)} - \frac{1}{\lambda^2(z)} + 3 z \frac{1}{\lambda(z)} > 0 \text{ for } z < 0. 
\]
Moving $\frac{1}{\lambda^2(z)}$ we obtain the following inequality:
\[
\frac{1}{\lambda^2(z)} \left[2 \lambda^2(z) +z^2 - 1 + 3 z \lambda(z) \right] > 0 \text{ for } z < 0.
\]
As $\lambda^2(z) > 0$, this inequality is equivalent to:
\[
2 \lambda^2(z) + z^2 - 1 + 3 z \lambda(z) > 0 \text{ for } z < 0.
\]
Moving two terms to the right side of the inequality we get the statement:
\[
2 \lambda^2(z) > 1 - z^2 - 3 z \lambda(z)  \text{ for } z < 0.
\]
\end{proof}

\paragraph{Proof of Theorem~\ref{theorem:sign_half_normal_e}}

\begin{proof}

Denote by $A = \frac{\mu_*}{\sigs}$.
As $A$ we have:
\begin{align*}
A &= -\varepsilon \frac{\sigu^2}{\sigu^2 + \sigv^2} \cdot \frac{\sqrt{\sigu^2 + \sigv^2}}{\sigu \sigv} = -\varepsilon \frac{\sigu}{\sigv} \frac{1}{\sqrt{\sigu^2 + \sigv^2}} = \\
&= -\varepsilon \frac{1}{\sigv^2} \frac{\sigu \sigv}{\sqrt{\sigu^2 + \sigv^2}} = -\varepsilon \frac{\sigs}{\sigv^2}.
\end{align*}

Using this notation we get:
\[
\bbE(u | \varepsilon) = \sigs \frac{\phi(A)}{\Phi(A)} + \sigs A = \sigs \left[\frac{\phi(A)}{\Phi(A)} + A \right].
\]

The desired partial derivative has the form:
\begin{align*}
    &\frac{\partial}{\partial \sigs} \bbE(u | \varepsilon) =
    \frac{\partial}{\partial \sigs} \left[\sigs \left(\frac{\phi(A)}{\Phi(A)} + A \right) \right] = 
    \frac{\partial}{\partial \sigs} \left[\sigs (\lambda(A) + A) \right] = \\
    &= \lambda(A) + A + \sigs (\lambda'(A) + 1) \frac{\partial A}{\partial \sigs} = 
    \lambda(A) + A + (1 + \lambda'(A)) \sigs \left( \frac{-\varepsilon}{\sigv^2} \right) = \\
    &= \lambda(A) + A + (1 + \lambda'(A)) A = \lambda(A) + 2 A + A \lambda'(A) = \\
    &= \frac{\phi(A) \Phi(A) + 2 A \Phi^2(A) + A (-A \phi(A) \Phi(A) - \phi^2(A)}{\Phi^2(A)} = \\
    &= \frac{1}{\Phi^2(A)} \left(\phi(A) \Phi(A) + 2 A \Phi^2(A) - A^2 \phi(A) \Phi(A) - A \phi^2(A) \right),
\end{align*}
as
\begin{align*}
\lambda'(z) &= \frac{\partial}{\partial z} \frac{\phi(z)}{\Phi(z)} = \frac{\phi'(z) \Phi(z) - \phi(z) \Phi'(z)}{\Phi^2(z)} = \frac{-z \phi(z) \Phi(z) - \phi^2(z)}{\Phi^2(z)} \\
&= -z\lambda(z) - \lambda^2(z),
\end{align*}
and
\[
\frac{\partial A}{\partial \sigs} = \frac{\partial}{\partial \sigs} \left(- \varepsilon \frac{\sigs}{\sigv^2} \right)
= -\frac{\varepsilon}{\sigv^2}. 
\]

So, to prove the theorem it is sufficient to prove that
\[
\forall z,\, \psi(z) = \phi(z) \Phi(z) + 2 z \Phi^2(z) - z^2 \phi(z) \Phi(z) - z \phi^2(z) > 0.
\]
It is equivalent to 
\begin{equation}
\label{eq:main_th}
\lambda(z) + 2z - z^2\lambda(z) - z\lambda^2(z) > 0.
\end{equation}
% or 
% \[
% \lambda(z) + 2z + z \lambda'(z) > 0.
% \]

We start with the case $z < 0$.

Multiplying the inequality by $2$ we get an equivalent inequality:
\[
2 \lambda(z) + 4 z - 2 z^2\lambda(z) - 2 z\lambda^2(z) > 0.
\]

From \eqref{eq:lambdaz} in Statement~\ref{statement:4} above:
\begin{align*}
2 \lambda(z) &+ 4 z - 2 z^2 \lambda(z) - 2 z \lambda^2(z) >
2 \lambda(z)  + 4 z - 2 z^2 \lambda(z) - z \left(1 - z^2 - 3 z \lambda(z)\right) \\ 
&= 2 \lambda(z) + 4 z - 2 z^2 \lambda(z) - z + z^3 + 3 z^2 \lambda(z) \\
&= 2 \lambda(z) + 3 z + z^2 \lambda(z) + z^3 = (2 + z^2) \lambda(z) + 3 z + z^3.  
\end{align*}

So, it is sufficient to prove, that for $z < 0$:
\begin{equation}
\label{eq:lambda_z_2}
(2 + z^2) \lambda(z) + 3 z + z^3 > 0.    
\end{equation}

From~(\cite{baricz2008mills}) we get that the following inequality holds:
\[
\frac{1}{\lambda(-x)} < \frac{4}{\sqrt{x^2 + 8} + 3 x}, x > 0.
\]
Using the change of variables $z = -x$ we get:
\[
\frac{1}{\lambda(z)} < \frac{4}{\sqrt{z^2 + 8} - 3 z}, z \leq 0.
\]
The exchange of nominator and denominator leads to:
\begin{equation}
\label{eq:baricz_bound}
\lambda(z) > \frac14 \left(\sqrt{z^2 + 8} -3 z \right), z \leq 0. \end{equation}

The inequality~\eqref{eq:lambda_z_2} is equivalent to:

\[
\lambda(z) > \frac{-3 z - z^3}{2 + z^2}.
\]

So, using the bound~\eqref{eq:baricz_bound} it is sufficient to prove, that for $z < 0$:
\[
\frac14 \left(\sqrt{z^2 + 8} -3 z \right) > \frac{-3 z - z^3}{2 + z^2}.
\]
For $x = -z \geq 0$ we get an equivalent inequality:
\[
\frac14 \left(\sqrt{x^2 + 8} + 3 x \right) > \frac{3 x + x^3}{2 + x^2}.
\]
Rearranging the terms we get the inequality:
\begin{equation}
\label{eq:inequality_sqrt}
\sqrt{x^2 + 8} > 4 \frac{3 x + x^3}{2 + x^2} - 3 x.    
\end{equation}

For the right side we have:
\[
4 \frac{3 x + x^3}{2 + x^2} - 3 x = \frac{12 x + 4 x^3 - 3 x^3 - 6 x}{2 + x^2} = \frac{x^3 + 6 x}{2 + x^2} = x + \frac{4 x}{x^2 + 2}.
\]

Both parts of~\eqref{eq:inequality_sqrt} are positive, so~\eqref{eq:inequality_sqrt} is equivalent to:
\[
x^2 + 8 > \left(x + \frac{4 x}{x^2 + 2} \right)^2.
\]
Moving $x^2$ to the right side we get:
\[
8 > \frac{8 x^2}{x^2 + 2} + \frac{16 x^2}{(x^2 + 2)^2}.
\]
Moving the first term at the right side to the left we get:
\[
8 \left( 1 - \frac{x^2}{x^2 + 2} \right) > \frac{16 x^2}{(x^2 + 2)^2}.
\]
Subtracting $\frac{x^2}{x^2 + 2}$ from $1$ we obtain:
\[
8 \frac{2}{x^2 + 2} > \frac{16 x^2}{(x^2 + 2)^2}.
\]

As
\[
1 > \frac{x^2}{x^2 + 2}.
\]
we proved~\eqref{eq:main_th} for $z < 0$.
The remaining part is the proof of~\eqref{eq:main_th} for $z > 0$.

As 
\[
-1 \leq \lambda'(z) \leq 0,
\]
we have:
\[
\lambda(z) + 2z - z^2\lambda(z) - z\lambda^2(z) = \lambda(z) + 2 z + z \lambda'(z) \geq \lambda(z) + 2 z + (-z) = \lambda(z) + z > 0.
\]
 
QED.
\end{proof}

\subsubsection{Proof of Theorem~\ref{theorem:sign_half_normal}}

\begin{proof}
\[
\bbE(e^{-u} | \varepsilon) = \frac{\Phi \left(\frac{\mu_*}{\sigs} - \sigs \right)}{\Phi \left(\frac{\mu_*}{\sigs} \right)} e^{-\mu_* + \frac12 \sigs^2}.
\]

For $A$ we have
\[
A = \frac{\mu_*}{\sigs} = -\varepsilon \frac{\sigs}{\sigv^2}.
\]
Then the partial derivative with respect to $\sigs$ has the form:
\[
\frac{\partial A}{\partial \sigs} = -\frac{\varepsilon}{\sigv^2}.
\]

For $E(e^{-u} | \varepsilon)$ we obtain:
\[
\bbE(e^{-u} | \varepsilon) = \frac{\Phi(A - \sigs)}{\Phi(A)} e^{-A\sigs + \frac12 \sigs^2}.
\]
Then the partial derivative has the form:
\begin{align*}
\frac{\partial \bbE(e^{-u} | \varepsilon)}{\partial \sigs} &= \frac{\partial}{\partial \sigs} \left[
\frac{\Phi(A - \sigs)}{\Phi(A)} \right] e^{-A\sigs + \frac12 \sigs^2} \\  
&+ \frac{\Phi(A - \sigs)}{\Phi(A)} e^{-A\sigs + \frac12 \sigs^2} \frac{\partial}{\partial \sigs}  \left(-A \sigs + \frac12 \sigs^2 \right).
\end{align*}

We continue to expand the terms above using in addition the following:
\[
-A \sigs + \frac12 \sigs^2 = \frac{\varepsilon \sigs^2}{\sigv^2} + \frac12 \sigs^2 = \sigs^2 \left( \frac12 + \frac{\varepsilon}{\sigv^2} \right).
\]

So,
\begin{align*}
\frac{\partial \bbE(e^{-u} | \varepsilon)}{\partial \sigs} &= 
\frac{1}{\Phi^2(A)} 
\Biggl(
 \phi(A - \sigs) \left(-\frac{\varepsilon}{\sigv^2} - 1 \right) \Phi(A) \\
 &- \Phi(A - \sigs) \phi(A) \left(- \frac{\varepsilon}{\sigv^2} \right)
\Biggr)
 e^{-A\sigs + \frac12 \sigs^2} + \\
&+ \frac{\Phi(A - \sigs)}{\Phi(A)} e^{\sigs^2 \left(\frac12 + \frac{\varepsilon}{\sigv^2} \right)} 2 \sigs \left(\frac12 + \frac{\varepsilon}{\sigv^2} \right) =\\
&= \frac{\Phi(A - \sigs)}{\Phi(A)} e^{\sigs^2 \left(\frac12 + \frac{\varepsilon}{\sigv^2} \right)}
\frac{1}{\sigs} \sigs \Bigg( \lambda(A - \sigs) \left(-\frac{\varepsilon}{\sigs^2} - 1 \right) - \\
&- \lambda(A) \left(- \frac{\varepsilon}{\sigv^2} \right) + 2  \sigs \left(\frac12 + \frac{\varepsilon}{\sigv^2} \right) \Bigg ).
\end{align*}

So, we need to prove that:
\[
\sigs \left(\lambda(A - \sigs) \left(-\frac{\varepsilon}{\sigs^2} - 1 \right) - \lambda(A) \left(- \frac{\varepsilon}{\sigv^2} \right) + 2  \sigs \left(\frac12 + \frac{\varepsilon}{\sigv^2} \right) \right) < 0.
\]
Or equivalently:
\[
\lambda(A - \sigs) (A - \sigs) - \lambda(A) A + \sigs^2 - 2 A \sigs < 0.
\]
If $x = A - \sigs$, then $A = x + \sigs = x + t, t > 0$ and we have:
\[
\lambda(x) x - \lambda(x + t) (x + t) + t^2 - 2 (x + t) t < 0
\]
Opening brackets we get:
\[
\lambda(x) x - \lambda(x + t) (x + t) - t^2 - 2 x t < 0
\]
So, we need to prove that for $t > 0$ and arbitrary $x$:
\[
\psi(x, t) = (x + t) \lambda(x + t) - x \lambda(x) + t^2 + 2 x t > 0.
\]
It holds that $\psi(x, 0) = 0$.
Then it is sufficient to prove that the function is increasing i.e. the corresponding partial derivative is positive:
\[
\frac{\partial \psi(x ,t)}{\partial t} = \lambda(x + t) + (x + t) \lambda'(x + t) + 2 t + 2 x > 0.
\]
Using the change of variables $z = x + t$ we get the inequality for $z \in (-\infty, +\infty)$:
\[
\lambda(z) + z \lambda'(z) + 2 z > 0.
\]
For $z > 0$ it is obvious that:
\[
z (1 + \lambda'(z)) + (z + \lambda(z)) > 0,
\]
as $0 < 1 + \lambda'(z) < 1$ and $z + \lambda(z) > 0$.

For $z < 0$ it is more complicated.
We need to prove, that for $z < 0$
\[
\lambda(z) + 2 z - z^2 \lambda(z) - z \lambda^2(z) > 0.
\]
Substituting $\lambda(z)$ by $\frac{\phi(z)}{\Phi(z)}$ we get:
\[
\phi(z) \Phi(z) + 2 z \Phi^2(z) - z^2 \phi(z) \phi(z) - z \phi^2(z) > 0.
\]
We apply the change of variables $x = -z$, so for $x > 0$ we want to prove:
\[
\phi(x) \Phi(-x) - 2 x \Phi^2(-x) - x^2 \phi(x) \Phi(-x) + x \phi^2(x) > 0.
\]
Let $F(x) = \Phi(-x)$.
Then we need to prove for $x > 0$:
\[
\phi(x) F(x) + 2 x F^2(x) - x^2 \phi(x) F(x) + x \phi^2(x) > 0.
\]
Rearranging terms we get the inequality:
\begin{equation}
\label{eq:main}
(1 - x^2) \phi(x) F(x) + x \phi(x)^2 - 2 x F(x)^2 > 0 \text{ for } x > 0,
\end{equation}
where $F(x) = 1 - \Phi(x)$.
To prove it we'll split the whole interval $(0, \infty)$ into two smaller ones: $(0, 1]$ and $(1, \infty)$.

\paragraph{$x \in (1, \infty)$}

In this case $1 - x^2 < 0$, and to prove~\eqref{eq:main} it is sufficient to prove:
\[
(1 - x^2) \frac{4}{\sqrt{x^2 + 8} + 3 x} + x - 2 x \frac{16}{(\sqrt{x^2 + 8} + 3 x)^2} > 0,
\]
as it holds that $F(x) \leq \frac{4}{\sqrt{x^2 + 8} + 3 x} \phi(x)$ according to (Baricz, 2007).

Then by multiplying by $(\sqrt{x^2 + 8} + 3 x)^2$ we get:
\begin{align*}
& 4 (1 - x^2) (\sqrt{x^2 + 8} + 3 x) + x (\sqrt{x^2 + 8} + 3 x)^2 - 32 x \\    
&= 4 \sqrt{x^2 + 8} + 12 x - 4 x^2 \sqrt{x^2 + 8} - 12 x^3 - 32 x \\
& \qquad + x(x^2 + 8 + 9 x^2 + 6 x \sqrt{x^2 + 8}) \\
&=  4 \sqrt{x^2 + 8} - 4 x^2 \sqrt{x^2 + 8} - 20 x - 12 x^3 + 10x^3 + 8 x + 6 x^2 \sqrt{x^2 + 8}  \\
&= 4 \sqrt{x^2 + 8} + 2 x^2 \sqrt{x^2 + 8} - 12 x - 2 x^3.
\end{align*}

So, we need to prove that: 
\begin{align*}
&4 \sqrt{x^2 + 8} + 2 x^2 \sqrt{x^2 + 8} > 12 x + 2 x^3 \Leftrightarrow \\
&\sqrt{x^2 + 8} (2 + x^2) > 6 x + x^3.
\end{align*}
As the left side and the right side of inequality are positive for $x > 0$
it is equivalent to the inequalities for the squares of both sides:
\begin{align*}
&(x^2 + 8) (2 + x^2)^2 > (6 x + x^3)^2 \Leftrightarrow \\
& (x^2 + 8) (4 + 4 x^2 + x^4) > 36 x^2 + 12 x^4 + x^6 \Leftrightarrow \\
& 4 x^2 + 4 x^4 + x^6 + 32 + 32 x^2 + 8 x^4 > 36 x^2 + 12 x^4 + x^6 \Leftrightarrow \\
& 32 + 36 x^2 + 12 x^4 + x^6 > 36 x^2 + 12 x^4 + x^6 \Leftrightarrow \\
& 32 > 0.
\end{align*}
We proved the inequality for the case $x > 1$.

\paragraph{Value of $x$ between $0$ and $1$}
%\paragraph{TODO}

We use the following strategy: we split to smaller intervals, for each interval we provide a bound $\phi(x) > c F(x)$ defined by the left edge of the interval as $\left(\phi(x) / F(x) \right)' > 0$ according to Lemma~\ref{Lemma0}, and then get a quadratic inequality or a linear inequality, which is easy to check.

Let's start with $x \in (0.9, 1]$.
$\phi(x) > 1.44 F(x)$, then 
\begin{align*}
&(1 - x^2) \phi(x) F(x) + x \phi(x)^2 - 2 x F^2(x) > \\
&(1 - x^2) 1.44 F^2(x) + 2.07 x F^2(x) - 2 x F^2(x) \geq \\
&2.07 x F^2(x) - 2 x F^2(x) > 0.07 x F^2(x) > 0.
\end{align*}

We proceed in a similar way for other intervals.
If $x \in (0.83, 0.9]$, then $\phi(x) > 1.39 F(x)$.
Then 
\begin{align*}
&(1 - x^2) \phi(x) F(x) + x \phi(x)^2 - 2 x F(x)^2 \\
& \qquad > 1.39 (1 - x^2) F(x)^2 + 1.93 x F(x)^2 - 2 x F(x)^2 \geq 0.
\end{align*}

If $x \in (0.65, 0.83]$, then $\phi(x) > 1.25 F(x)$.
Then 
\begin{align*}
&(1 - x^2) \phi(x) F(x) + x \phi(x)^2 - 2 x F(x)^2 \\
& \qquad > 1.25 (1 - x^2) F(x)^2 + 1.5625 x F(x)^2 - 2 x F(x)^2 \geq 0.
\end{align*}

If $x \in (0.4, 0.65]$, then $\phi(x) > 1.05 F(x)$.
Then 
\begin{align*}
&(1 - x^2) \phi(x) F(x) + x \phi(x)^2 - 2 x F(x)^2 \\
& \qquad > 1.05 (1 - x^2) F(x)^2 + 1.1025 x F(x)^2 - 2 x F(x)^2 \geq 0.
\end{align*}

If $x \in [0, 0.4]$, then $\phi(x) > 0.75 F(x)$.
Then 
\begin{align*}
&(1 - x^2) \phi(x) F(x) + x \phi(x)^2 - 2 x F(x)^2 \\
& \qquad > 0.75 (1 - x^2) F(x)^2 + 0.5625 x F(x)^2 - 2 x F(x)^2 \geq 0.
\end{align*}

QED.
\end{proof}

\subsection{Proof of the Statement~\ref{statement:effect_te_sign}}

\begin{proof}
We consider a discrete random variable $u$.
It takes values $u_i = z\, u_{i0}, \, i=1, 2$ with probabilities $p_1, p_2$
correspondingly, where $u_{i0} > 0, \, i=1, 2$.
Since $v\sim \mathcal{N}(0, \sigv^2)$ and $u$ are independent and 
$\varepsilon = v - u$, the joint distribution of $u, \varepsilon$ has the form
\[
f(u = u_i, \varepsilon) = p_i \frac{1}{\sqrt{2 \pi} \sigv} \exp \left(-\frac{(u_i + \varepsilon)^2}{2 \sigv^2} \right).
\]
Thus, the marginal pdf of $\varepsilon$ has the form:
\begin{equation}
\label{eq:density_eps}
f(\varepsilon) = \sum_{i = 1}^2 p_i \frac{1}{\sqrt{2 \pi} \sigv} \exp \left(-\frac{( u_i + \varepsilon)^2}{2 \sigv^2} \right).
\end{equation}
The conditional distribution has the form:
\[
P(u = u_i | \varepsilon) = \frac{p_i \frac{1}{\sqrt{2 \pi} \sigv} \exp \left(-\frac{( u_i + \varepsilon)^2}{2 \sigv^2} \right)}
{\sum_{i = 1}^2 p_i \frac{1}{\sqrt{2 \pi} \sigv} \exp \left(-\frac{( u_i + \varepsilon)^2}{2 \sigv^2} \right)}
=\frac{p_i e^{-w_i} } {p_1 e^{-w_1}+p_2 e^{-w_2}}, \,\,\, i=1, 2,
\]
where $w_i = \frac{( u_i + \varepsilon)^2} {2 \sigv^2} = \frac{(z\, u_{i0} + \varepsilon)^2} {2 \sigv^2}$.

Then observation-specific technical efficiency is 
\begin{align}
TE &= \bbE \left(e^{-u} | \varepsilon\right) 
= \sum_{i = 1}^2 { e^{-u_i} \frac{p_i e^{-w_i} } {\sum_{j = 1}^2 p_j e^{-w_j}} }
=  \frac{ \sum_{i = 1}^2 p_i  e^{-u_i} e^{-w_i} } {\sum_{j = 1}^2 p_j e^{-w_j}}
\notag \\
&=  \frac{ \sum_{i = 1}^2 p_i  e^{-z \,u_{i0}} e^{-w_i} } {\sum_{j = 1}^2 p_j e^{-w_j}}. \label{eq:TE_discrete}
\end{align}

Then the marginal effect $\frac{\partial TE}{\partial z}$ equals:
\begin{align*}
\frac{\partial TE}{\partial z} 
&= \frac{1}{\sum_{j = 1}^2 p_j e^{-w_j}}
\frac{\partial }{\partial z} \sum\nolimits_{i = 1}^2 p_i  e^{-z \,u_{i0}} e^{-w_i}\\
&- \frac{1}{ \left(\sum_{j = 1}^2 p_j e^{-w_j} \right)^2}
\left( \sum\nolimits_{i = 1}^2 p_i  e^{-z \,u_{i0}} e^{-w_i} \right)
\frac{\partial }{\partial z} \sum\nolimits_{j = 1}^2 p_j e^{-w_j}\\
&= - \frac{1}{\sum_{j = 1}^2 p_j e^{-w_j}} 
  \sum\nolimits_{i = 1}^2 p_i  e^{-z \,u_{i0}} e^{-w_i} (u_{i0}+ w'_i)\\
&+ \frac{1}{ \left(\sum_{j = 1}^2 p_j e^{-w_j} \right)^2}
\left( \sum\nolimits_{i = 1}^2 p_i  e^{-z \,u_{i0}} e^{-w_i} \right)
  \left( \sum\nolimits_{j = 1}^2 p_j e^{-w_j} w'_j \right),
\end{align*}
where $w'_i =\frac{\partial}{\partial z}w_i   = \frac{z\, u_{i0}^2 + \varepsilon\, u_{i0}} { \sigv^2}$
\end{proof}

\section{Identifiability of the normal-discrete model}
\label{sec:identifiability_check}

We examined the discrete model in a number of ways.
The most important issue to check was identifiability of the model.

We use the dataset of size $1000$, generated with the normal-discrete model,
which we used for Fig.\ref{fig:hat_TE_against_TE} in Section~\ref{sec:discrete_distr}. 
We use the maximum likelihood approach with  p.d.f. from~\eqref{eq:density_eps} 
to estimate the normal-discrete model. Estimated $\widehat{TE}_i$ for this model were calculated  from~\eqref{eq:TE_discrete}. Also for this data we estimated two 
misspecified models: normal-half normal and normal-exponential and derived predicted
technical efficiencies $\widehat{TE}_i$ for these models. 
Figures~\ref{fig:hat_TE_against_TE_three_models} contain comparison of true values of $TE$ and their three estimates $\widehat{TE}_i$ using three different models.
%The dataset of size $1000$ were generated with the normal-discrete model, defined %in Section~\ref{sec:discrete_distr}.
%The p.d.f. from~\eqref{eq:density_eps} was used during data generation.
%Then we use the maximum likelihood approach of the normal-discrete model or %misspecified the normal-half-normal and the normal-exponential model for parameter %and TE estimatations. 
%Figures~\ref{fig:hat_TE_against_TE_three_models} contain comparison of true values %of TE and their estimates using three different models.
We see that if the model is correctly specified, obtained estimates are close to the real ones.
While, if we start to use common, but misspecified normal-half-normal and normal-exponential models, 
the estimates are worse.

\begin{figure}[h]
	\centering
	\begin{subfigure}[b]{0.485\textwidth}
		\includegraphics[width=\textwidth]{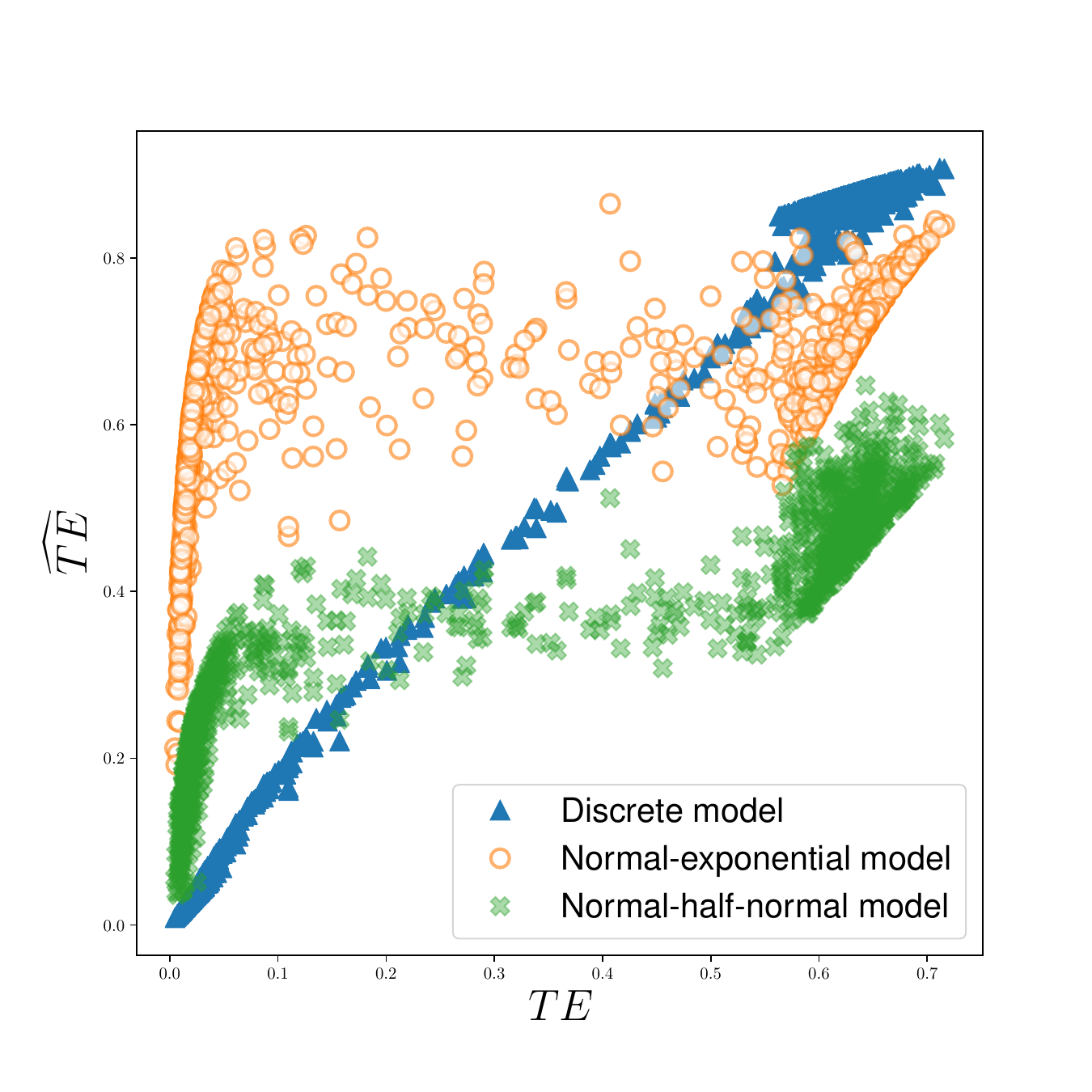}
		\caption{Comparison of all estimates of $\widehat{TE}$ and true $TE$}
	\end{subfigure}
	~ %add desired spacing between images, e. g. ~, \quad, \qquad, \hfill etc. 
	%(or a blank line to force the subfigure onto a new line)
	\begin{subfigure}[b]{0.485\textwidth}
		\includegraphics[width=\textwidth]{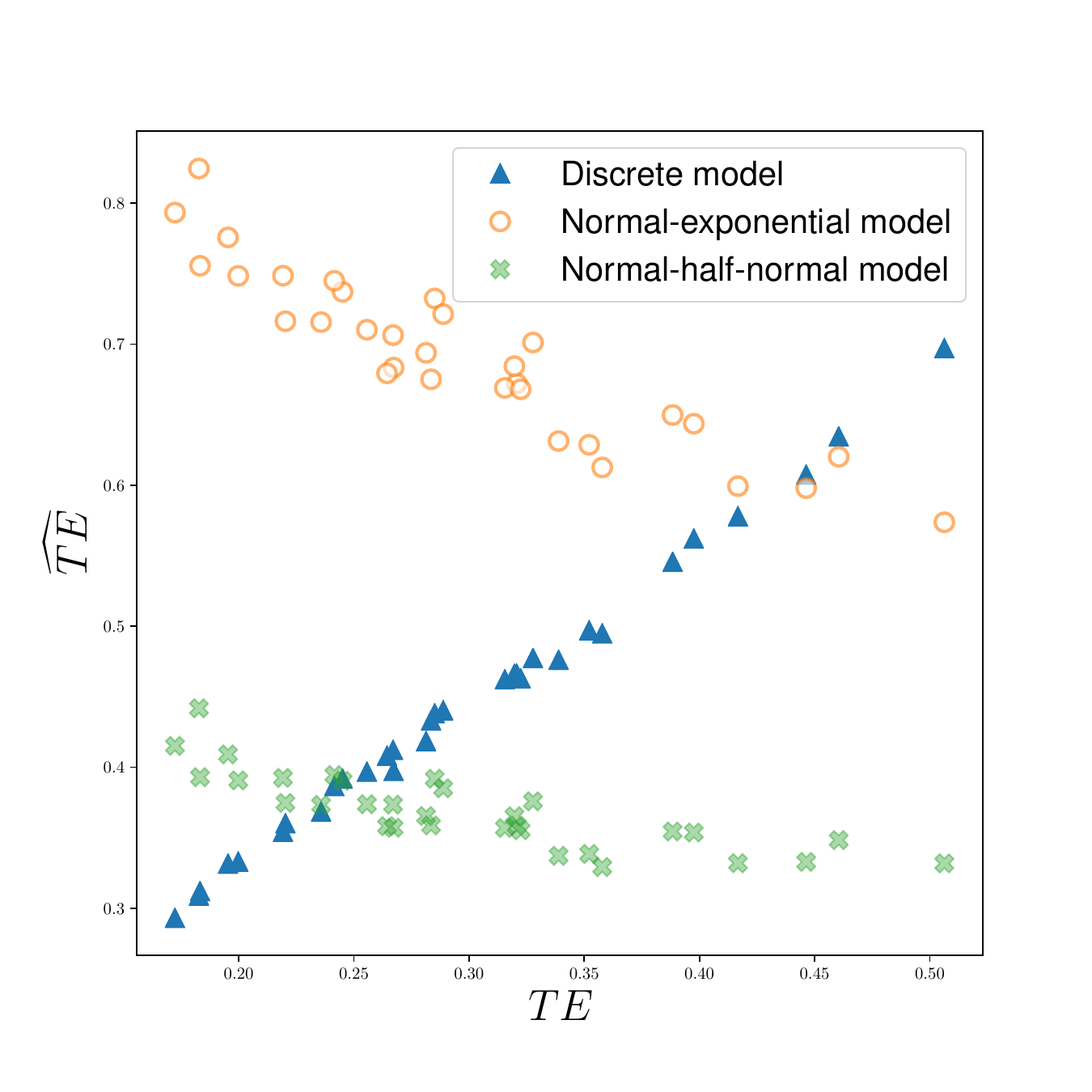}
		\caption{Selected points with $\varepsilon$ between $-2.3$ and $-2.1$}
	\end{subfigure}
	\caption{Comparison of estimates $\widehat{TE}$ using a normal-discrete, a normal-half-normal and a normal-exponential models and true $TE$ obtained using a normal-discrete model}	
	\label{fig:hat_TE_against_TE_three_models}
\end{figure}

\begin{table}[h]
	\centering
	\begin{tabular}{lcc}
		\hline
		Model     & Correlation & Correlation  \\
		&             & $-2.3 < \varepsilon < -2.1 $ \\
		\hline
		Normal-discrete        & $0.9816$ & $\phantom{-}0.9971$ \\
		Normal-half-normal     & $0.9451$  & $-0.8768$  \\
		Normal-exponential     & $0.7616$ & $-0.9285$ \\
		\hline     
	\end{tabular}
	\caption{Spearman rank correlations for true values and the three estimates of $TE$ if the true model is normal-discrete}
	\label{tab:quality_estimates_three_models}
\end{table}

Spearman rank correlation between true $TE$ and the three predicted $\widehat{TE}$
are 
%To make it more formal we can measure correlations between true values and estimates obtained using three different models.
%Spearman-rank correlations 
provided in Table~\ref{tab:quality_estimates_three_models}.
The highest rank correlation is obtained when the true model is estimated. 
The correlation is smaller for the for the normal-half-normal model and is even worse for the normal-exponential model. But for the subset of observations selected by the condition  $-2.3 < \varepsilon < -2.1 $ both misspecified models provide strongly negative rank correlations of predicted  $\widehat{TE}$ and true values of the technical efficiency $TE$.

\section*{References}

\bibliographystyle{apa} 
\bibliography{article.bib}

\end{document}